\documentclass[english, 11pt]{article}
\usepackage[margin=1in]{geometry}
\usepackage[final]{microtype}
\usepackage{hyperref}
\usepackage{mathtools}
\usepackage{amssymb, amsmath, amsthm}
\usepackage[capitalise]{cleveref}
\usepackage[utf8]{inputenc} 
\usepackage[backend=bibtex,maxbibnames=99,sortcites]{biblatex}
\usepackage{enumitem}

\renewbibmacro*{doi+eprint+url}{
    \printfield{doi}
    \newunit\newblock
    \iftoggle{bbx:eprint}{
        \usebibmacro{eprint}
    }{}
    \newunit\newblock
    \iffieldundef{doi}{
        \usebibmacro{url+urldate}}
        {}
    }
\hypersetup{
    colorlinks,
    linkcolor={red!50!black},
    citecolor={blue!50!black},
    urlcolor={blue!80!black}
}
\usepackage[font={small}]{caption}
\usepackage{xcolor}

\newcommand{\norm}[1]{\lVert #1\rVert}
\newcommand{\lrc}[1]{\left \{#1\right\}}
\newcommand{\R}{\mathbb{R}}
\newcommand{\CQ}{\mathcal{Q}}
\newcommand{\T}{\mathcal{T}}
\newcommand{\opt}{{\textsc{opt}}}

\DeclareMathOperator{\peri}{perimeter}

\DeclareMathOperator{\area}{area}
\DeclareMathOperator{\poly}{poly}
\DeclareMathOperator{\diam}{diam}
\newcommand{\pparagraph}[1]{\paragraph*{#1}}
\makeatletter
\newtheorem*{rep@theorem}{\rep@title}
\newcommand{\newreptheorem}[2]{
\newenvironment{rep#1}[1]{
 \def\rep@title{#2 \ref{##1}}
 \begin{rep@theorem}}
 {\end{rep@theorem}}}
\makeatother

\newtheorem{theorem}{Theorem}
\newreptheorem{theorem}{Theorem}

\newtheorem{lemma}[theorem]{Lemma}
\newreptheorem{lemma}{Lemma}

\theoremstyle{definition}

\theoremstyle{remark}
\newtheorem*{remark}{Remark}
\newtheorem*{claim}{Claim}

\crefname{theorem}{Theorem}{Theorems}
\crefname{lemma}{Lemma}{Lemmas}
\crefname{proposition}{Proposition}{Propositions}
\crefname{claim}{Claim}{Claims}
\crefname{corollary}{Corollary}{Corollaries}
\crefname{definition}{Definition}{Definitions}
\crefname{figure}{Figure}{Figures}
\crefname{equation}{Equation}{Equations}
\crefname{section}{Section}{Sections}
\crefname{remark}{Remark}{Remarks}
\title{Partitioning a Polygon Into Small Pieces}
\author{Mikkel Abrahamsen\thanks{Department of Computer Science, University of Copenhagen, \texttt{miab@di.ku.dk}.
Supported by Starting Grant 1054-00032B from the Independent Research Fund Denmark under the Sapere Aude research career programme and is part of Basic Algorithms Research Copenhagen (BARC), supported by the VILLUM Foundation grant 16582.} \and Nichlas Langhoff Rasmussen\thanks{Department of Computer Science, University of Copenhagen, \texttt{nichlas.rasmussen@gmail.com}.}}
\date{October 2024}

\begin{document}

\thispagestyle{empty}

\maketitle

\begin{abstract}
We study the problem of partitioning a given simple polygon $P$ into a minimum number of connected polygonal pieces, each of bounded size.
We describe a general technique for constructing such partitions that works for several notions of `bounded size,' namely that each piece must be contained in an axis-aligned or arbitrarily rotated unit square or a unit disk, or that each piece has bounded perimeter, straight-line diameter or geodesic diameter.
The problems are motivated by practical settings in manufacturing, finite element analysis, collision detection, vehicle routing, shipping and laser capture microdissection.

The version where each piece should be contained in an axis-aligned unit square is already known to be NP-hard~[Abrahamsen and Stade, FOCS, 2024], and the other versions seem no easier.
Our main result is to develop constant-factor approximation algorithms,
which means that the number of pieces in the produced partition is at most a constant factor larger than the cardinality of an optimal partition.
Existing algorithms~[Damian and Pemmaraju, Algorithmica, 2004] do not allow Steiner points, which means that all corners of the produced pieces must also be corners of $P$.
This has the disappointing consequence that a partition often does not exist, whereas our algorithms always produce meaningful partitions.
Furthermore, an optimal partition without Steiner points may require $\Omega(n)$ pieces for polygons with $n$ corners where a partition consisting of just $2$ pieces exists when Steiner points are allowed.
Other existing algorithms~[Arkin, Das, Gao, Goswami, Mitchell, Polishchuk and T{\'{o}}th, ESA, 2020] only allow $P$ to be split along chords (and aim to minimize the number of chords instead of the number of pieces), whereas we make no constraints on the boundaries of the pieces.

In a related problem, we are given a polygon $P$ and positive real values $a_1,\ldots,a_k$ whose sum $\sum_{i=1}^k a_i$ equals the area of $P$.
The goal is to partition $P$ into exactly $k$ pieces $Q_1,\ldots,Q_k$ such that the area of $Q_i$ is $a_i$.
Such a partition always exists, and an algorithm with running time $O(nk)$ has previously been described [Bast and Hert, CCCG, 2000].
We improve on this result and give an algorithm with optimal running time $O(n+k)$ for simple polygons and a running time of $O(n\log n+k)$ for polygons with holes.
\end{abstract}

\newpage

\clearpage
\setcounter{page}{1}

\section{Introduction}
Consider a manufacturing process where we want to produce a large object $P$, but we are only able to produce parts of bounded size.
We then partition $P$ into parts that are sufficiently small to be fabricated, and in the end these can be assembled to form the desired object $P$.
To keep both fabrication and assembly as simple as possible, we naturally want as few parts as possible.
In this paper we provide the first algorithms that can split $P$ into parts of bounded size, where the number of parts is provably close to the minimum possible.

To model the object $P$, we use a simple polygon.
Polygons are among the geometric structures that are most widely used to model physical objects, as they are suitable for representing a wide variety of shapes and figures in computer graphics and vision, pattern recognition, robotics, computer-aided design and manufacturing, and other computational fields.
Our partitioning problems belong to the class of \emph{decomposition problems}, which form an old and large sub-field in Computational Geometry.
In all of these problems, we want to \emph{decompose} a polygon $P$ into connected polygonal \emph{pieces} that must respect a certain restriction.
Here, the union of the pieces should be $P$, and we usually seek a decomposition into as few pieces as possible in which case the decomposition is called \emph{optimal}.
A decomposition where the pieces can overlap is also called a \emph{cover}, and if they are pairwise interior-disjoint, it is called a \emph{partition}.
Depending on the assumptions about the input polygon $P$ and the requirements to the pieces, this leads to a wealth of interesting problems.
There is a vast literature about such decomposition problems, as documented in several highly-cited books and survey papers that give an overview of the state-of-the-art at the time of publication~\cite{shermer1992recent, chazelle1994decomposition, chazelle1985approximation, keil1999polygon, keil1985minimum, o1987art,o2004polygons}.

However, as it turns out, there are no known algorithms with provable guarantees for partitioning a polygon into few small pieces, when no further restrictions are made on the partition.
A \emph{Steiner point} is a corner of a piece in a decomposition that is not a corner of the input polygon~$P$.
We describe algorithms for constructing the following six types of partitions with Steiner points allowed:

\begin{itemize}

\item Aligned square partition: Each piece is contained in an axis-aligned unit square.

\item Rotated square partition: Each piece is contained in an arbitrarily rotated unit square.

\item Disk partition: Each piece is contained in a unit disk, i.e., a disk of radius $1$.

\item Straight diameter partition: The straight-line diameter of each piece is at most $1$.

\item Geodesic diameter partition: The geodesic diameter (i.e., the maximum length of a shortest path) of each piece is at most $1$.

\item Perimeter partition: The perimeter of each piece is at most $1$.
\end{itemize}

The aligned square partitioning problem was recently shown to be NP-hard~\cite{DBLP:journals/corr/abs-2404-09835}; this is the first partitioning problem known to be hard even for simple polygons with the usual objective of minimizing the number of pieces.
The other variants are probably just as hard, but do not lend themselves to an NP-hardness reduction as easily due to their more complicated nature.
In fact, it seems intractable to find optimal solutions even in extremely restricted settings:
if $P$ is a (big) square or an equilateral triangle, it seems out of reach to efficiently compute optimal rotated square, disk, diameter (straight as well as geodesic) and perimeter partitions of $P$.
Researchers have tried to determine the largest squares and equilateral triangles that have disk partitions with $k$ pieces, for various fixed values of $k$, but for squares it is already unknown for $6$ pieces~\cite{heppes1997covering} and for equilateral triangles, it is unknown for $7$~\cite{melissen1997loosest}.\footnote{The papers~\cite{heppes1997covering,melissen1997loosest} ask for the largest square or equilateral triangle that can be \emph{covered} with $k$ unit disks, but one can prove that this is equivalent to finding the largest for which a disk \emph{partition} exists with $k$ pieces.}
Likewise, the rotated square partitioning problem is closely related to the problem of covering a bigger square with a minimum number of unit squares, which has received some attention in mathematics~\cite{DBLP:journals/tamm/Januszewski09, DBLP:journals/jct/Soifer06}.
It is, for instance, not known if $6$ unit squares can cover a square of size larger than $2\times 2$~\cite{DBLP:journals/tamm/Januszewski09}.
See also the webpage~\cite{ErickPack} for an overview of the best-known solutions to these and similar problems.

We therefore develop approximation algorithms instead, i.e., algorithms that produce partitions of cardinality at most a constant factor $\alpha$ larger than that of an optimal partition.
We use $\opt$ to denote the cardinality of an optimal partition of a particular type.
We distinguish between \emph{estimation time} and \emph{construction time}.
Estimation time is the time it takes to compute a number $x$ such that $\opt\leq x\leq \alpha \cdot\opt$.
Construction time is the time it takes to construct an actual partition with at most $\alpha\cdot \opt$ pieces.
Note that $\Omega(n+\opt)$ time is needed for construction, and $\opt$ can be superpolynomial in $n$, for instance if $P$ is just a big square.
In contrast to many other partitioning problems, such as triangulation, trapezoidation or partitioning into convex pieces, it is therefore not possible to solve the construction problem in $O(\poly n)$ time.
Our results can be seen in \Cref{table:results}.

\begin{table}
\centering
\begin{tabular}{|l|l|l|l|}
\hline
Type & Apx.~factor & Estimation time & Construction time \\
\hline
\hline
Aligned square & $13$ & $O(n)$ & $O(n^2+\opt\log n)$ \\
\hline
Rotated square & $21$ & $O(n\log n)$ & $O(n^2+\opt\log n)$ \\
\hline
Disk & $20+\pi/2$ & $O(n\log n)$ & $O(n^2+\opt\log n)$ \\
\hline
Straight diameter & $20+\pi/2$ & $O(n\log n)$ & $O(n^2+\opt\log n)$ \\
\hline
Geodesic diameter & $72$ & $O(n^2\log n)$ & $O(n^2\log n+\opt\log n)$ \\
\hline
Perimeter & $3728$ & $O(n^2)$ & $O(n^2+\opt\log n)$ \\
\hline
\end{tabular}
\caption{Approximation factors and running times of our algorithms for solving the various partitioning problems.
Recall that $\opt$ denotes the cardinality of an optimal partition of the respective type and $n$ denotes the number of corners of $P$.}
\label{table:results}
\end{table}

\begin{theorem}\label{thm:main-2}
For aligned and rotated square, disk, straight and geodesic diameter, and perimeter partitions, there exist algorithms with running times $O(\poly n)$ for computing a constant factor estimate of $\opt$ and algorithms with running times $O(\poly n+\opt\log n)$ for constructing a partition with at most that cardinality.
The approximation factors and running times for each type are stated in \Cref{table:results}.
\end{theorem}

\begin{figure}
\centering
\includegraphics[page=18]{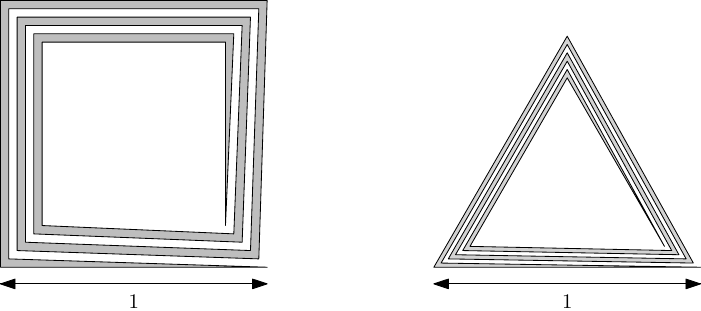}
\caption{The spirals indicate that we cannot use an algorithm for aligned square partitions or straight diameter partitions to get an algorithm for the other problem.}
\label{fig:EquiTri}
\end{figure}

Let us note that a $O(1)$-approximation algorithm for one of our problems does not seem to be useful (as a black box) in order to get a $O(1)$-approximation algorithm for another; see \Cref{fig:EquiTri}.
A polygon of straight-line diameter at most $1$ is also contained in an aligned unit square, but a straight diameter partition of the left spiral has $\Omega(n)$ pieces while an aligned square partition of one piece exists.
On the other hand, a polygon contained in an aligned square of side length $1/\sqrt 2$ has straight diameter at most $1$, but a partition into aligned $1/\sqrt 2$-squares of the right spiral needs $\Omega(n)$ pieces while a straight diameter partition of one piece exists.

A variant of the disk partitioning problem has been studied before when Steiner points are not allowed in the produced partition.
Damian and Pemmaraju~\cite{damian2004computing} described a polynomial-time algorithm to produce an optimal disk partition without using Steiner points.
Unfortunately, such a partition does often not exist.
For instance, no disk partition exists if $P$ is a triangle with an edge of length more than $2$.
The algorithm is therefore of little use in many situations. 

\begin{figure}
\centering
\includegraphics[page=1]{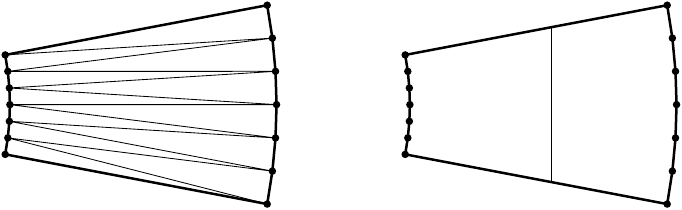}
\caption{Two disk partitions of the same polygon $P$.
The corners of $P$ are contained in two concentric circular arcs.
The concave arc has radius $1+\varepsilon$ for a small value $\varepsilon>0$.
The convex arc has radius $3$.
When Steiner points are not allowed, where $\Omega(n)$ pieces are needed in a disk partition.
With Steiner points, $2$ pieces are enough.}
\label{fig:Steiner}
\end{figure}

Another advantage of our partitions is that they are often significantly smaller than partitions without Steiner points.
\cref{fig:Steiner} shows that $\Omega(n)$ pieces are sometimes needed when Steiner points are not allowed for polygons where $2$ pieces are enough when they can be used.


An additional and very natural variant of the problem would be to partition $P$ into pieces of area at most $1$.
In fact, a partition into the optimal number $\lceil \area P \rceil$ of such pieces always exists.
We study the following more general \emph{area partitioning problem}:
Given a polygon $P$ and positive real values $a_1,\ldots,a_k$ with $\area P =\sum_{i=1}^k a_i$, compute a partition of $P$ into pieces $Q_1,\ldots,Q_k$ such that $\area Q_i=a_i$.
We show that such a partition can be found in optimal time $O(n+k)$ if $P$ is a simple polygon and in time $O(n\log n+k)$ if $P$ has holes.
This improves on a result by Bast and Hert~\cite{bast2000area}, who gave an algorithm with running time $O(nk)$.
Let us mention that in the context of manufacturing, pieces of bounded area can hardly be considered small, as they may have arbitrarily large diameter.

\begin{theorem}\label{thm:main-1}
There exists an algorithm that for a given simple polygon $P$ with $n$ corners and positive real values $a_1,\ldots,a_k$ with $\area P =\sum_{i=1}^k a_i$ produces a partition $Q_1,\ldots,Q_k$ with $\area Q_i =a_i$ in optimal time $O(n+k)$.
If $P$ is a polygon with holes, a partition can be found in time $O(n\log n+k)$.
\end{theorem}

\subsection{Other practical motivation}

Our partitioning problems are motivated by various practical domains besides manufacturing, where it is necessary to partition large objects into few objects of bounded size, possibly with additional constraints on the pieces that we ignore in this paper.
Let us mention finite element analysis~\cite{GONZALEZHERRERA2005337,MCCLUNG1989253,MCCLUNG1989237,tomar2004bounds,tusnina2014}, collision detection~\cite{damian2004computing}, guarding and servicing~\cite{hert1998polygon,DBLP:journals/talg/CarlssonAY10}, shipping, laser capture microdissection~\cite{selbach2021shape}.

\subsection{Other related work}

Motivated by indoor localization using tripwire lasers, Arkin, Das, Gao, Goswami, Mitchell, Polishchuk and T{\'{o}}th~\cite{DBLP:conf/esa/ArkinD0GMPT20} recently studied problems of the type where we want to cut a polygon into small pieces using a minimum number of chords of $P$ (each chord corresponds to a laser that can detect if something is blocking the beam).
Fekete, Kamphans, Kr{\"{o}}ller, Mitchell and Schmidt~\cite{fekete2011exploring} (see also the video~\cite{DBLP:conf/compgeom/BeckerFKLMS13}) studied the problem of finding a Steiner triangulation of a simple polygon where each edge has length at most $1$.
Motivated by the construction of communication graphs for robots, they wanted to minimize the number of vertices and gave a $3$-approximation algorithm.

Worman~\cite{worman2003decomposing} showed that it is NP-hard to find an optimal aligned square partition of a polygon \emph{with holes} when Steiner points are not allowed by a reduction from Planar 3,4-SAT.
Buchin and Selbach~\cite{buchin2021decomposing} recently showed that the same holds for disk partitions by a similar reduction.

The algorithmic problems of partitioning polygons into pieces of bounded size have older roots in pure mathematics.
For instance, in the \emph{Borsuk problem}, we consider a convex body $P$ in $\mathbb R^d$ of diameter $1$ and want to partition $P$ into a minimum number of pieces, each of which has diameter less than $1$.
This problem has a rich history; see~\cite{borsuk1933drei,DBLP:journals/combinatorics/JenrichB14,kahn1993counterexample} for a few references.
In \emph{Conway's fried potato problem}~\cite{bezdek1995solution, bezdek1996conway,canete2022conway,croft2012unsolved}, we seek to minimize the
maximum in-radius of each piece after a given number of successive cuts by hyperplanes for a given convex polyhedron in $\mathbb R^d$.

Specialized versions of the area partitioning problem have also been studied~\cite{adjiashvili2010equal,DBLP:journals/talg/CarlssonAY10,hert1998polygon}.
Other problems related to the area partitioning problem are the \emph{equipartition problems}, where we seek to partition a convex polygon into convex pieces all having the same area or the same perimeter or both (or other measures)~\cite{DBLP:journals/tcs/ArmaseluD15, blagojevic2014convex,guardia2005equipartition,karasev2014convex,nandakumar2012fair}.

\subsection{Technical overview}\label{sec:techniques}

\begin{figure}
\centering
\includegraphics[page=30]{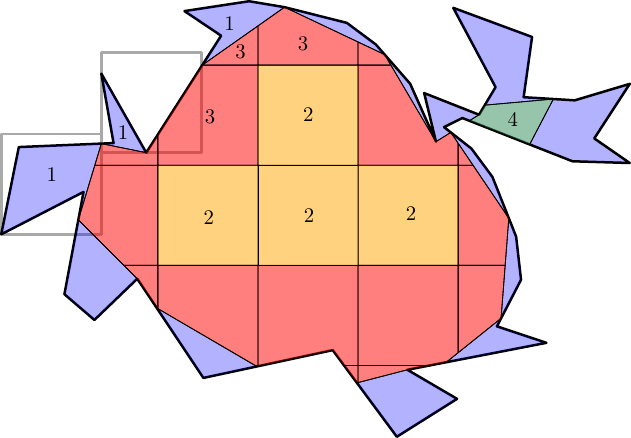}
\caption{An aligned square partition.
The boundary pieces are blue, and they cover all of the boundary, although they are in some places degenerate so that they cannot be seen.
The yellow, red and green pieces are so-called complete, edge and chip pieces, respectively. Color codes: 1 blue, 2 yellow, 3 red, 4 green.
}
\label{fig:SquarePartition}
\end{figure}

For all the variants, except the area partitioning problem, we use the same general technique that consists of two main steps; see \Cref{fig:SquarePartition} for an example of an aligned square partition produced by our algorithm.

\begin{enumerate}
\item \label{item:step1}
We first construct a set of pieces $\CQ_\partial$ that cover all of the boundary of $P$, as described in \Cref{sec:boundary-pieces}.
For each variant, we describe in \Cref{sec:maxintervals} a greedy algorithm that repeatedly finds a maximal interval on the boundary of $P$ that can be contained in a single piece.
\Cref{sec:constructingpieces} describes how we turn these intervals into actual pieces, and the method depends on the specific problem:
For square\footnote{When we write about ``square  partitions'' without specifying whether it is aligned or rotated square partitions, we mean both types.}, disk and straight diameter partitions, we ``blow up'' a piece as a balloon, starting from its interval on the boundary of $P$, until we reach either the boundary of $P$, a previously constructed piece or the boundary of the convex hull of the interval.
For geodesic diameter and perimeter partitions, it does not work to blow up pieces like that, as it can result in pieces of too large geodesic diameter or perimeter.
Instead, we connect the endpoints of the interval with a shortest path in $P$, and the piece is the region enclosed by the union of the interval and this path.

These pieces are called \emph{boundary} pieces.
As we will see, constructing the pieces in a greedy manner results in at most $2\,\opt$ pieces (recall that $\opt$ denotes the cardinality of an optimal partition of the entire polygon $P$).

\item \label{item:step2}
It remains to construct the \emph{interior} pieces $\mathcal Q_\circ$, i.e., the pieces that do not meet $P$ at the boundary, and this is described in
\Cref{sec:interiorpieces}.
To this end, we consider a regular square grid that covers all of $P$, where the edge length of each square is some constant $\gamma>0$ to be chosen depending on the problem.
This grid is used to construct the \emph{interior pieces} $\mathcal Q_\circ$.
Square, disk and straight diameter partitions are handled in \Cref{sec:DSSDinterior}.
Here, the interior pieces are simply chosen as the regions in each grid square that are inside $P$ but not inside a boundary piece.
An interior piece is either \emph{complete}, i.e., it is a complete square, or \emph{incomplete}, in which case it shares boundary with a boundary piece.
By an area argument, there can be $O(\opt)$ complete pieces.
Due to the construction of the boundary pieces $\CQ_\partial$, an incomplete piece has a corner that is an intersection point between a grid line and the boundary of the convex hull of a boundary piece.
There can be $O(1)$ such intersection points per boundary piece, and they can be used to account for the incomplete pieces.
Since there are $O(\opt)$ boundary pieces, we then get $O(\opt)$ interior pieces, and thus $O(\opt)$ pieces in total.

\Cref{sec:interiorgeodesic} describes the construction for geodesic diameter and perimeter partitions, and it is more involved due to two different issues:
1) A region in a square can have arbitrarily large geodesic diameter or perimeter.
For instance, a region can have the shape of a long spiral and thus violate the size constraint.
2) There can be too many regions.
For instance, a single boundary piece can have a spiralling behavior, crossing back and forth between two neighbouring grid squares and thus creating arbitrarily many small regions in those squares.
To address the first issue, we develop a method to split each region into subregions, each of which has bounded geodesic diameter and perimeter:
A region is split by choosing a fixed point in the region and cutting along shortest paths to all points on the boundary where two boundary pieces meet, as well as some additional points.
Each of the resulting subregions will be bounded by three concave chains, so it has bounded geodesic diameter and perimeter.
To address the second issue, we define our pieces as certain unions of these subregions, which may be contained in different but neighbouring grid squares.
That makes it possible to bound the number of pieces as $O(\opt)$ using bounds on the number of boundary pieces and the number of squares in the grid.
\end{enumerate}

For the area partitioning problem, we first construct a \emph{Hamiltonian triangulation} of the polygon $P$ using Steiner points.
This is a triangulation where the dual graph contains a Hamiltonian cycle $C$~\cite{arkin1996hamiltonian}.
The Hamiltonian triangulation is obtained from a usual triangulation by adding the medians of all the triangles.
We then construct the pieces of our partition in a greedy manner as we traverse the cycle $C$.
\Cref{fig:AreaPartition} shows an example of the triangulation and a partition produced by the algorithm. The details are given in \Cref{sec:area}.

\begin{figure}
\centering
\includegraphics[page=21]{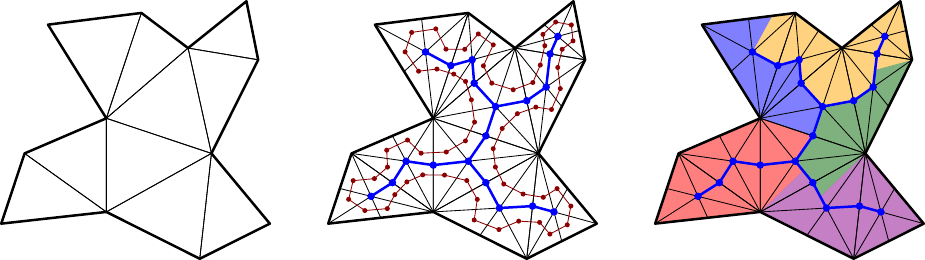}
\caption{Left: A triangulated polygon.
Middle: The triangulation we obtain after adding all medians of the triangles.
The Hamiltonian cycle $C$ is also shown.
Right: An area partition into five pieces produced by our algorithm.}
\label{fig:AreaPartition}
\end{figure}

\section{Preliminaries}

A \emph{simple polygon} is a compact region in the plane whose boundary is a simple, closed curve consisting of finitely many line segments.
For technical reasons, we allow the pieces of a partition to be weakly simple polygons. 
A \emph{weakly simple polygon} $Q$ is a simply-connected and compact region in the plane whose boundary is a union of finitely many line segments.
In particular, a simple polygon is also a weakly simple polygon, but the opposite is not true in general.
For instance, a weakly simple polygon $Q$ may have a disconnected or even empty interior.
However, just as for a simple polygon, a weakly simple polygon $Q$ can be defined by its edges in counterclockwise order around the boundary.
These edges form a closed boundary curve $\lambda$ of $Q$.
Since $Q$ is weakly simple, some corners may coincide, and edges may overlap.
After a perturbation of $\lambda$ that is arbitrarily small with respect to Frechet distance, $Q$ can be turned into a simple polygon~\cite{DBLP:journals/dcg/AkitayaAET17,DBLP:conf/soda/ChangEX15}.
This perturbation may involve the introduction of more corners.
For instance, if $Q$ is just a line segment, then $Q$ has only two corners, and one more is needed to obtain a simple polygon.

We denote the boundary of a weakly simple polygon $P$ as $\partial P$.
For two boundary points $a,b\in\partial P$, we denote by $\partial P[a,b]$ the interval on the boundary from $a$ to $b$ in counterclockwise direction.
If $P$ is not simple, then the interval $\partial P[a,b]$ is in general not well-defined, since the boundary $\partial P$ can meet itself several times at $a$ or $b$.
However, we will only use the notation when it is clear from the context which interval we mean.
As with intervals on the real number line, we can use square or round brackets to indicate whether each of the endpoints $a$ and $b$ are included in the interval or not.
A boundary interval is \emph{convex} if it contains no concave corner of $P$ in its interior, and \emph{concave} if it contains no convex corner in its interior.

We now define a \emph{partition} of a polygon $P$ to be a set of weakly simple polygons $Q_1,\ldots,Q_k$ such that after an arbitrarily small perturbation of the pieces, we obtain \emph{simple} polygons $Q'_1,\ldots,Q'_k$ with the following properties:

\begin{enumerate}
    \item The polygons $Q'_1,\ldots,Q'_k$ are pairwise interior-disjoint. \label{partprop:1}
    
    \item $\bigcup_{i=1}^k Q'_i =P$. \label{partprop:2}
\end{enumerate}

Note that it follows that the weakly simple polygons $Q_1,\ldots,Q_k$ must also have properties \ref{partprop:1} and \ref{partprop:2} (with $Q'_i$ replaced by $Q_i$ for all $i\in\{1,\ldots,k\}$), since otherwise a large perturbation would be needed for them to be transformed into simple polygons with the required properties.
However, it would not be sufficient to define a partition as a set of weakly simple polygons with properties \ref{partprop:1} and \ref{partprop:2} alone. This would, for instance, allow two pieces with empty interiors (such as two segments) to properly intersect each other, which is not intended.

For a weakly simple polygon $P$, we define the \emph{straight-line diameter} (or simply \emph{straight diameter}) of $P$ as $\diam P = \sup_{p,q \in P} \norm{p-q}$.
For two points $p,q\in P$, we denote by $\pi(p, q)=\pi_P(p, q)$ a shortest path in $P$ from $p$ to $q$.
We define the \emph{geodesic diameter} of $P$ as $\sup_{p,q \in P} \Vert\pi_P(p,q)\Vert$, where $\Vert\cdot\Vert$ denotes the length.
Note that $\pi(p, q)$ is a polygonal curve whose corners, except possibly the endpoints, are concave corners of $P$.

\section{Computing the boundary pieces}\label{sec:boundary-pieces}

In this and the next section, we turn our attention to square, disk, straight and geodesic diameter and perimeter partitions in order to prove~\Cref{thm:main-2}.
A \emph{boundary partition} of a polygon $P$ is a set of weakly simple polygons $Q_1,\ldots, Q_k$ of bounded size such that after an arbitrarily small perturbation of the pieces, we obtain \emph{simple} polygons $Q'_1,\ldots, Q'_k$ with the following properties:

\begin{enumerate}
\item The polygons $Q'_1,\ldots, Q'_k$ are pairwise interior-disjoint.

\item Each polygon $Q'_i$ is contained in $P$.

\item $\partial P\subset \bigcup_{i=1}^k \partial Q'_i$.

\item $\partial P\cap \partial Q'_i$ is connected for each polygon $Q'_i$.
\end{enumerate}

In particular, it follows that $\partial P\cap \partial Q_i$ may be disconnected for a piece $Q_i$ in a boundary partition, but after an arbitrarily small perturbation, we get a piece $Q'_i$ where  $\partial P\cap \partial Q'_i$ is connected, so other pieces must ``take over'' covering the parts of $\partial P$ that $\partial Q_i$ covers but $\partial Q'_i$ does not.
\Cref{fig:boundaryperturb} demonstrates the perturbation.
See also~\Cref{fig:SquarePartition}, where the blue pieces form a boundary partition produced by one of our algorithms.
We require the pieces $Q_1,\ldots,Q_k$ of a boundary partition to have bounded size, but the perturbed pieces $Q'_1,\ldots,Q'_k$ might violate the size constraint.

\begin{figure}
\centering
\includegraphics[page=3]{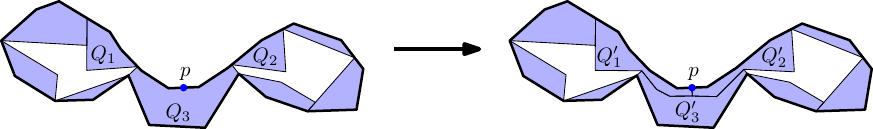}
\caption{The blue pieces form a boundary partition of a polygon $P$.
The pieces $Q_1$ and $Q_2$ are degenerate and follow the boundary of $P$ to the point $p$.
It is shown how we can make a small perturbation so that the intersection of each piece with $\partial P$ is just a single interval.}
\label{fig:boundaryperturb}
\end{figure}

Note that more than $\opt$ pieces may be needed to create a boundary partition because of the last requirement that $\partial P\cap \partial Q'_i$ be connected.
However, as the following lemma shows, at most $2\,\opt -2$ are needed.

\begin{lemma}\label{lem:size-boundary-partition}
For each type of partition mentioned in \cref{thm:main-2} and any simple polygon $P$, there exists a boundary partition of $P$ of cardinality at most $2\,\opt-2$.
\end{lemma}

\begin{proof}
Consider any of the six types of partitions and let $\CQ^{*}$ denote an optimal partition of $P$ of that type so that $|\CQ^{*}| = \opt$.
Let $\CQ^{*}_{b} \subset \CQ^{*}$ denote the set of pieces of $\CQ^{*}$ that intersect $\partial P$, i.e., 
\[
\CQ^{*}_{b} = \lrc{Q \in \CQ^{*} \mid \partial Q \cap \partial P \neq \emptyset}.
\]

Let $C$ denote the directed cycle with vertex set $\CQ^{*}_{b}$ and edges induced by traversing $\partial P$ in counterclockwise order:
Whenever we leave a piece $Q$ and enter another piece $Q'$, we add an edge from $Q$ to $Q'$ in $C$; see Figure~\ref{fig:boundarysplit}.
(If some pieces of $\CQ^{*}_{b}$ are degenerate weakly simple polygons, this cycle $C$ may not be well-defined, but then we can consider an infinitesimal perturbation of $\CQ^{*}_{b}$ instead, resulting in pairwise interior-disjoint simple polygons that we use to define $C$.)
We now observe that $|E(C)| \leq 2m - 2$ by induction on the number of pieces $m=|\CQ^{*}_{b}|$.
The base-case is $|\CQ^{*}_{b}|=1$, where the claim holds trivially.
Suppose that the claim holds for any cycle of at most $m-1$ pieces, and consider a cycle $C$ of $m=|\CQ^{*}_{b}|$ pieces.
If $C$ is a simple cycle of $m\geq 2$ pieces, then $|E(C)|=m\leq 2m-2$.
Otherwise, there is a piece $Q$ that is visited at least twice as we traverse $C$.
We can then write $C=C_1\cup C_2$, where $C_1$ and $C_2$ are two cycles that are disjoint except that $Q$ appears in both, so we can write the numbers of pieces as $m_1$ and $m_2$, respectively, where $m=m_1+m_2-1$.
We then get from the induction hypothesis that $|E(C)|=|E(C_1)|+|E(C_2)|\leq 2m_1-2+2m_2-2=2(m_1+m_2-1)-2=2m-2$.

\begin{figure}
\centering
\includegraphics[page=25]{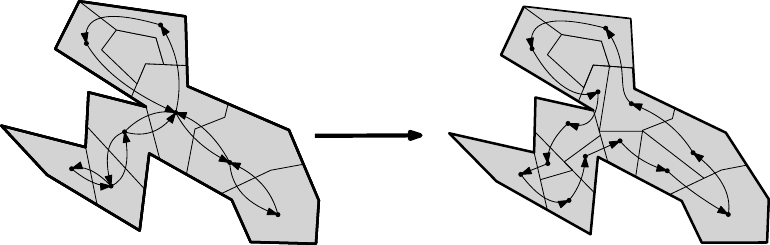}
\caption{We split the pieces of an optimal partition (left) along shortest paths in order to obtain a partition that includes a boundary partition (right).
For clarity, the endpoints of the shortest paths are not chosen as described in the text.}
\label{fig:boundarysplit}
\end{figure}

We now partition each piece $Q\in\CQ^{*}_{b}$ into $d=\deg^-(Q)=\deg^+(Q)$ pieces as follows, where $\deg^-(\cdot )$ and $\deg^+(\cdot )$ denote the in- and out-degree, respectively.
Note that $d$ is equal to the number of intervals in $\partial P\cap \partial Q$.
Denote these intervals as $I_1,\ldots,I_d$.
For $i\in\{1,\ldots,d-1\}$, we split $Q$ along a shortest path in $Q$ between the endpoints of $I_i$.
It is then clear that each of the resulting $d$ pieces has bounded size and that we obtain a boundary partition of $P$.

We can then bound the number of pieces in our boundary partition as $\sum_{Q \in \CQ^{*}_{b}} \mathrm{deg}^-(Q) = |E(C)| \leq 2m-2\leq 2\,\opt - 2$, and the lemma holds. 
\end{proof}

All of our algorithms construct the boundary pieces in a greedy manner, as follows.
We pick an arbitrary point $a_0\in\partial P$.
For $i=1,2,\ldots$, we then choose the next point $a_i$ such that the interval $\partial P[a_{i-1},a_i]$ is as large as possible while it can be covered by a single piece $Q_i$ satisfying the respective size constraint.
We stop when we get to a point $a_{k-1}$ such that $\partial P[a_{k-1},a_0]$ can be covered by a piece $Q_k$ of bounded size.
Due to the greedy choices, the following lemma easily follows from~\Cref{lem:size-boundary-partition}.
As an upper bound on $|\CQ_\partial|$, we have $2\,\opt-1$ instead of $2\,\opt-2$, since we choose an arbitrary starting point $a_0$, which can result in an excess of one boundary piece.

\begin{lemma}\label{lemma:boundarybound}
Let $a_0=a_k$ be an arbitrary point on $\partial P$ and let $\CQ_\partial=\{Q_1,\ldots,Q_k\}$ be a set of pieces of bounded size where $Q_i$ covers the interval $\partial P[a_{i-1},a_i]$ and for $i=1,\ldots,k-1$, the endpoint $a_i$ is chosen so that a piece of bounded size cannot cover a larger interval starting at $a_{i-1}$.
Then $k=|\CQ_\partial| \leq 2\,\opt-1$.
\end{lemma}

\begin{proof}
By \cref{lem:size-boundary-partition} there exists a boundary partition, $\CQ^{*}_{\partial}$, of cardinality at most $2\,\opt-2$.
Since the interval $\partial P(a_{i-1},a_i]$ is chosen to be maximal, it must contain the last point of an interval covered by one of the pieces in $\CQ^{*}_{\partial}$.
In other words, there is at least one piece in $\CQ^{*}_{\partial}$ for each of the pieces $Q_1,\ldots,Q_{k-1}$, and we have $|\CQ_\partial| \leq 2\,\opt-1$.
\end{proof}

Next, we first describe how to find the maximal intervals and then how we construct the actual pieces.

\subsection{Finding maximal intervals}\label{sec:maxintervals}
In the following we show for each type of partition how to obtain a single maximal interval $\partial P[a_{i-1}, a_i]$ for a given starting point $a_{i-1}$.
The remaining intervals are computed by applying this algorithm repeatedly until we get to a point $a_{k-1}$ where $\partial P[a_{k-1}, a_0]$ can be covered by a piece of bounded size.
The point $a_0$ is chosen as an arbitrary point on $\partial P$. 
For most types of partitions, the algorithm consists of two phases.

\begin{enumerate}
\item \label{int:phase1} Given $a_0$ and $a_{i-1}$, consider the interval $\partial P[a_{i-1},a_0]$ that remains to be covered, and let $\lrc{p_1, \ldots, p_m}$ denote the corners of this interval in order from $a_{i-1}$, so that $p_1=a_{i-1}$ and $p_m=a_0$.
In this first phase we determine the largest index $s\in\{2,\ldots,m+1\}$ such that there exists a piece of bounded size that contains the interval $\partial P[p_1,p_{s-1}]$.
If $s=m+1$, it means the remaining part of $\partial P$ can be covered by a single piece, and we have found all our intervals and proceed to construct the pieces as described in \Cref{sec:constructingpieces}.
For each type of partition, we describe how to find the index $s$.
    
\item \label{int:phase2} The maximum interval $\partial P[a_{i-1},a_i]$ that can be covered by a single piece must consist of the interval $\partial P[p_1,p_{s-1}]$ and a prefix of the edge $p_{s-1}p_s$.
In this second phase we find the point $a_i \in p_{s-1}p_{s}$ such that $p_{s-1}a_i$ is this prefix.
\end{enumerate}

\pparagraph{Aligned square partitions.}
We describe the two phases for finding the maximal intervals for the aligned square partitions.

\begin{enumerate}
\item To find the index $s$, we run through the points $\{p_1,\ldots,p_m\}$ and keep track of the smallest and largest $x$- and $y$- coordinates, which define the smallest containing axis-aligned square.
The index $s$ is thus found in $O(s)$ time.

\item Let $R$ be the minimum axis-aligned bounding rectangle of the interval $\partial P[p_1,p_{s-1}]$.
There is an axis-aligned unit square $S$ covering a maximum interval of $\partial P[a_{i-1},a_0]$ that also has a corner coincident with the corresponding corner of $R$ (for instance, if the $x$- and $y$-coordinates of $a_i$ are at least as large as those of $a_{i-1}$, we can choose $S$ with the same lower left corner as $R$).
There are four unit squares with this property, and we just need to test which one covers most of $p_{s-1}p_s$; see \Cref{fig:cases-square}.
\end{enumerate}

In total, the point $a_i$ can be found in $O(s)$ time.

\begin{figure}
\centering
\includegraphics[page=22]{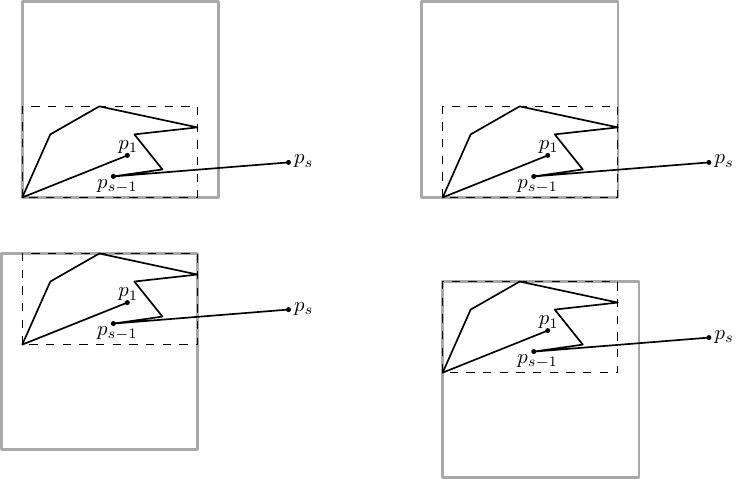}
\caption{We test four unit squares in order to find one that covers most of the edge $p_{s-1}p_s$.
}
\label{fig:cases-square}
\end{figure}

\pparagraph{Rotated square partitions.}
We describe the two phases for finding the maximal intervals for rotated square partitions.

\begin{enumerate}
\item \label{item:1rotatedsquare}
We perform an exponential search among the points $\{p_1,\ldots,p_m\}$ in order to find the index $s$.
In each round of the search, we need to test if an interval $\partial P[p_1,p_j]$ can be covered by a rotated square.
To this end, we first compute the convex hull $\mathcal C$ of $\partial P[p_1,p_j]$, which can be done in $O(j)$ time, for instance using Melkmann's algorithm~\cite{DBLP:journals/ipl/Melkman87}.
Toussaint~\cite{toussaint1983solving} described an algorithm that runs through the angle interval $\varphi\in[0,\pi/2)$ and at all times keep track of the smallest rectangle $R_\varphi$ containing $\mathcal C$ and with a pair of edges making the angle $\varphi$ with the $x$-axis.
The algorithm uses the technique of rotating calipers and the running time is $O(j)$.
Let $\ell_\varphi$ be the length of the longest edges of $R_\varphi$ (it may change several times which pair of edges are longest).
Thus, if ever $\ell_\varphi\leq 1$, there exists a unit square containing $\mathcal C$, otherwise not.
We can therefore determine in $O(j)$ time if there exists a unit square containing $\partial P[p_1,p_j]$ and use in total $O(s\log s)$ time to find $s$.
Another algorithm is described in~\cite{DBLP:journals/tcs/BeregB0KMS18}.

\begin{figure}
\centering
\includegraphics[page=27]{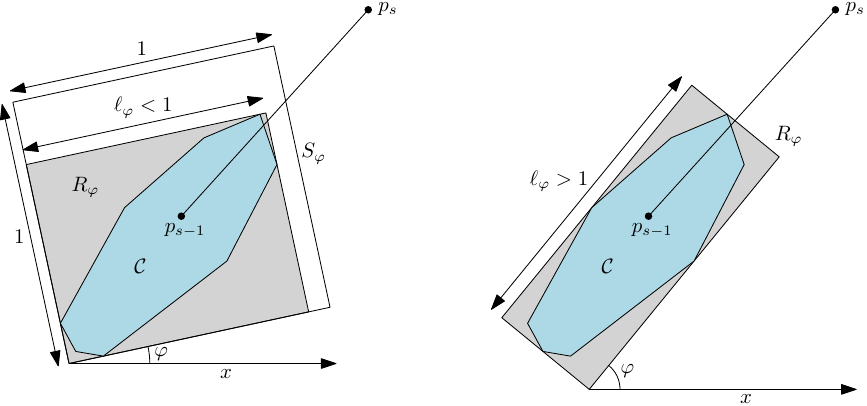}
\caption{Left: In some angle intervals, we have $\ell_\varphi\leq 1$, so there exists a unit square containing $\mathcal C$ and with the same orientation as $R_\varphi$.
We then keep track of the unit square $S_\varphi$ that covers most of $p_{s-1}p_s$, and this square shares a corner with $R_\varphi$, so there are only four candidates.
Right: In other intervals, we have $\ell_\varphi> 1$, so there is no unit square containing $\mathcal C$ with that orientation.}
\label{fig:rotatingsquares}
\end{figure}

\item In order to find out how much of the edge $p_{s-1}p_s$ we can cover with a unit square that also covers $\partial P[p_1,p_{s-1}]$, we can extend the algorithm described under point~\ref{item:1rotatedsquare}; see \Cref{fig:rotatingsquares}.
While running through the whole interval $\varphi\in [0,\pi/2)$, we keep track of the angle subintervals where $\ell_\varphi\leq 1$.
Whenever this is the case, we monitor how much of the edge $p_{s-1}p_s$ we can cover with a unit square aligned with $R_\varphi$.
As in the case of axis-aligned unit squares, the optimal such unit square $S_\varphi$ has a corner coincident with the corresponding corner of $R_\varphi$, so we just need to keep track of four unit squares.
We can therefore find the optimal square over all angles in $O(s)$ time.
\end{enumerate}
In total, the point $a_i$ can be found in $O(s\log s)$ time.

\pparagraph{Disk partitions.}
We describe the two phases for finding maximal intervals for disk partitions.

\begin{enumerate}
\item 
We perform an exponential search among the points $\{p_1,\ldots,p_m\}$ to find the index $s$.
At each step we can use an algorithm by Megiddo~\cite{DBLP:journals/siamcomp/Megiddo83a} to test if an interval $\partial P[p_1,p_j]$ is contained in a unit disk in $O(j)$ time.
The index $s$ can thus be found in $O(s\log s)$ time.

\item
The \emph{circular hull} of a set $M \subset \R^2$ is the intersection of all unit disks that contain $M$.
The circular hull is convex and if $M$ is finite, the boundary of the circular hull is a finite union of circular arcs of unit radius.
Let $\mathcal C$ be the circular hull of $\partial P[p_1,p_{s-1}]$, which is the same as the circular hull of the points $\{p_1,\ldots,p_{s-1}\}$.
Consider the unit disk $D$ that contains $\partial P[p_1,p_{s-1}]$ and a maximum prefix $p_{s-1}a_i$ of $p_{s-1}p_s$.
The following claim states that the disk $D$ falls under one of two possible cases that are depicted in \Cref{fig:disk-rotation}.

\begin{figure}
\centering
\includegraphics[page=23]{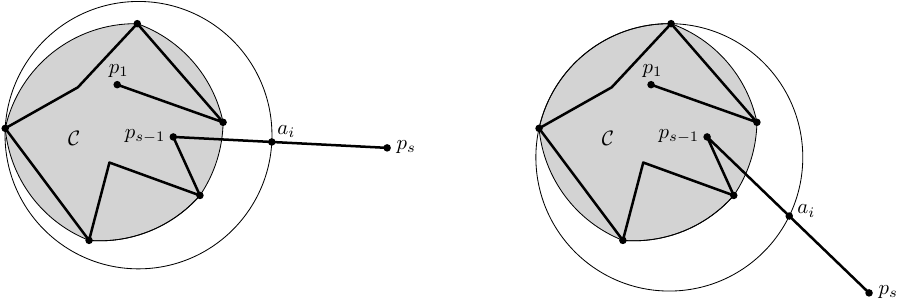}
\caption{The figure shows two cases of a unit disk covering a maximum prefix of $\partial P[p_1,p_s]$.
The circular hull of $\partial P[p_1,p_{s-1}]$ is in gray.}
\label{fig:disk-rotation}
\end{figure}

\begin{claim}
If there exists a disk $D'$ such that $\mathcal C\subset D'$ and $\partial D'$ contains a vertex $v$ of $\partial \mathcal C$ and the segment $va_i$ is a diameter of $D'$, then there is only one such disk $D'$ and we have $D=D'$.
Otherwise, $\partial D$ contains a unit radius arc from $\partial \mathcal C$.
\end{claim}

\begin{proof}
Clearly a disk $D'$ with the described properties is the unique disk $D$ that covers most of $p_{s-1}p_s$, as no other disk can cover more of $p_{s-1}p_s$ while also containing $v$.

If such a disk $D'$ does not exist, then the optimal disk $D$ must anyway contain a vertex $v$ of $\partial \mathcal C$, since otherwise it could be moved to cover more of $p_{s-1}p_s$.
Since $va_i$ is not a diameter of $D$, rotating $D$ either counterclockwise or clockwise while keeping $v$ a fixed point on $\partial D$ will result in a disk covering more of $p_{s-1}p_s$.
In the former case (rolling counterclockwise) it must therefore be the case that the unit radius arc after $v$ in counterclockwise order is contained in $\partial D$, while in the latter, the arc before $v$ is in $\partial D$.
\end{proof}

The claim suggests an algorithm to find the optimal disk $D$.
We first compute the circular hull $\mathcal C$ using an algorithm by Edelsbrunner, Kirkpatrick and Seidel~\cite{edelsbrunner1983shape}, which runs in time $O(s \log s)$.
We let $D'$ be a disk whose boundary contains a unit radius arc on $\partial\mathcal C$ and can then informally ``roll'' $D'$ around $\mathcal C$, keeping track of when $D'$ covers most of $p_{s-1}p_s$.
To be precise, let $A_1,\ldots,A_m$ be the unit radius arcs on $\partial\mathcal C$ in counterclockwise order, and let $v_1,\ldots,v_m$ be the vertices such that $v_i$ is the last point in $A_i$ in counterclockwise direction.
We first define $D'$ to be the disk such that $A_i\subset\partial D'$.
We then rotate $D'$ in counterclockwise direction while keeping $v_1$ on the boundary until we obtain $A_2\subset \partial D'$.
We then rotate around $v_2$ instead and continue in this way until we get back to $A_1$.
While rotating around a vertex $v_i$, we check if we eventually cover a point $q\in p_{s-1}p_s$ such that $v_iq$ is a diameter of $D'$, in which case we know that we found the optimal disk and can stop.
Otherwise, we keep track of the disk containing an arc $A_i$ that covers most of $p_{s-1}p_s$, which must be the optimal disk.
We can therefore find the optimal disk in $O(s)$ time once we know the circular hull $\mathcal C$.
\end{enumerate}

In total, the point $a_i$ can be found in $O(s\log s)$ time.

\pparagraph{Straight diameter partitions.}
We describe the two phases for finding the maximal intervals for the straight diameter partitions. 
\begin{enumerate}
\item
We use exponential search among the points $\{p_1,\ldots,p_m\}$ to find the index $s$.
In order to compute the straight diameter of an interval $\partial P[p_1,p_j]$, we compute the convex hull $\mathcal C$ of this interval, which can be done in $O(j)$ time, just as we did in the first phase for rotated square partitions.
We then compute the straight diameter of $\mathcal C$ using the technique of the rotating calipers~\cite[Theorem 4.19]{Preparata1985} in time proportional to the complexity of $\mathcal C$.

\item
In order to compute the maximal prefix of the edge $p_{s-1}p_s$ that can be covered by a piece of bounded diameter, we first construct the farthest point Voronoi diagram $\mathcal V$ of the points $\{p_1,\ldots,p_{s-1}\}$, which can be done in $O(s\log s)$ time~\cite[pp.~252--253]{Preparata1985}; see \Cref{fig:straightdiameter}.
The diagram $\mathcal V$ is a subdivision of the plane into convex regions $V_1,\ldots,V_{s-1}$ such that $V_j$ consists of all points $q$ where $p_j$ is the point among $\{p_1,\ldots,p_{s-1}\}$ with the largest distance to $q$.
We now traverse the segment $p_{s-1}p_s$ and keep track of the region $V_j$ in which we are.
When the distance to the farthest point $p_j$ reaches $1$, we stop and have found the point $a_i$.
Since $\mathcal V$ has complexity $O(s)$ and the regions are convex, this traversal can be done in $O(s)$ time.
\end{enumerate}
In total, the point $a_i$ can be found in $O(s\log s)$ time.

\begin{figure}
\centering
\includegraphics[page=24]{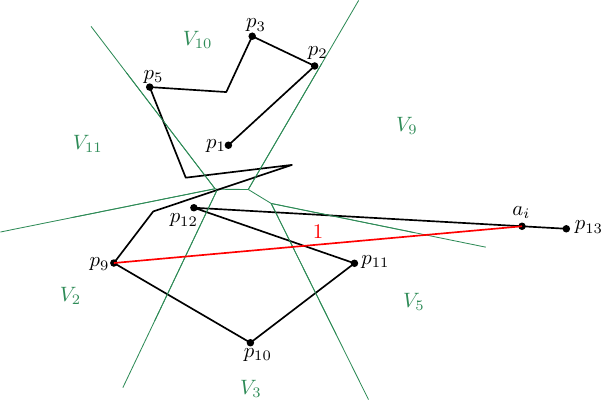}
\caption{In this example, we have $s=13$.
The farthest point Voronoi diagram of the vertices $\{p_1,\ldots,p_{12}\}$ is shown.
In order to find the maximal prefix of the interval $\partial P[p_1,p_{13}]$ of diameter at most $1$, we traverse the edge $p_{12}p_{13}$ and thus pass through the regions $V_2,V_3,V_5,V_9$ and finally get to the point $a_i$ where the distance to the farthest point $p_9$ is $1$.}
\label{fig:straightdiameter}
\end{figure}

\pparagraph{Geodesic diameter partitions.}
We describe the two phases for finding the maximal intervals for the geodesic diameter partitions.
\begin{enumerate}
\item
We use exponential search among the points $\{p_1,\ldots,p_m\}$ to find the index $s$.
In order to test if there exists a piece of bounded geodesic diameter covering a given interval $\partial P[p_1,p_j]$, we first construct the weakly simple polygon $Q_j$ bounded by the union of $\partial P[p_1,p_j]$ and the shortest path $\pi(p_1,p_j)$ from $p_1$ to $p_j$.
To do so we use an algorithm by Lee and Preparata~\cite{DBLP:journals/networks/LeeP84} to compute $\pi(p_1,p_j)$ in $O(n)$ time, using Chazelle's algorithm~\cite{chazelle1991triangulating} for triangulation in $O(n)$ time as preprocessing. We then compute the geodesic diameter of $Q_j$ using an algorithm by Hershberger and Suri~\cite{DBLP:journals/siamcomp/HershbergerS97} in $O(n)$ time (note that $Q_j$ might have complexity $\Omega(n)$ since many corners of $P$ that are not on the interval $\partial P[p_1,p_j]$ can be corners of $Q_j$).
We therefore use $O(n\log s)$ time to find the index $s$.

\item
This phase is similar to the second phase for straight diameter partitions, using a geodesic farthest point Voronoi diagram instead of a usual one:
In order to find the maximal prefix of the edge $p_{s-1}p_s$ that can be covered by a piece that also covers the interval $\partial P[p_1,p_{s-1}]$, we compute the geodesic farthest point Voronoi diagram $\mathcal V$ of the corners $p_1,\ldots,p_{s-1}$ in $P$.
This can be done in $O(n+s\log s)$ time using an algorithm by Wang~\cite{wang2022optimal}.
We then traverse the edge $p_{s-1}p_s$ and keep track of the region of $\mathcal V$ in which we are, and we stop when we reach the point $a_i$ where the geodesic distance to the farthest point will exceed $1$ by any further movement.
This traversal takes $O(n)$ time.
\end{enumerate}
In total, the point $a_i$ can be found in $O(n\log n)$ time.

\begin{figure}
\centering
\includegraphics[page=28]{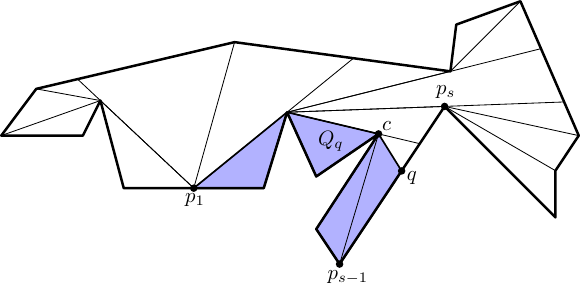}
\caption{We traverse the boundary $\partial P$ with the point $q$ and stop when the perimeter of $Q_q$ is $1$.
The shortest path tree is shown.
}
\label{fig:greedyperimeter}
\end{figure}

\pparagraph{Perimeter partitions.}
When constructing perimeter partitions, we can do both phases at once, as follows.
We first compute the \emph{shortest path tree} $T$ from the point $p_1$.
The tree $T$ partitions $P$ into regions such that for all points $q$ in one region, the shortest path $\pi(p_1,q)$ from $p_1$ to $q$ passes through the same sequence of corners of $P$.
The tree $T$ can be constructed in $O(n)$ time using an algorithm by Guibas, Hershberger, Leven, Sharir, and Tarjan~\cite{DBLP:journals/algorithmica/GuibasHLST87} and using Chazelle's algorithm~\cite{chazelle1991triangulating} for triangulation in $O(n)$ time as preprocessing.
We then traverse the interval $\partial P[p_1,p_m]$ with a moving point $q$; see \Cref{fig:greedyperimeter}.
We keep track of the perimeter of the weakly simple polygon $Q_q$ bounded by the union of the interval $\partial P[p_1,q]$ and the shortest path $\pi(p_1,q)$.
Note that this perimeter is monotonically increasing since the length of $\partial P[p_1,q]$ increases with unit speed (if $q$ moves with unit speed) while the length of $\pi(p_1,q)$ can decrease, but not with more than unit speed.
We stop when we reach a point $q=a_i$ where the perimeter of $Q_q$ is $1$ and any further movement of $q$ would make the perimeter exceed $1$.
We use the shortest path tree to keep track of the length of the path $\pi(p_1,q)$; the region containing $q$ defines the last corner $c$ of $P$ on the path $\pi(p_1,q)$, and the length of $\pi(p_1,q)$ can thus be computed as the sum of the precomputed length of the path $\pi(p_1,c)$ and the length of the segment $cq$.
In total, the point $a_i$ can be found in $O(n)$ time.

\subsection{Constructing boundary pieces from the intervals}\label{sec:constructingpieces}
Define the interval $I_i=\partial P[a_{i-1}, a_i]$ and let $\mathcal I$ be the set of the intervals $I_i$.
For the disk, square and straight diameter partitions, we define $C_i$ to be the convex hull of $I_i$ and $\mathcal C$ to be the set of these polygons $C_i$.
We then iterate through the intervals $I_i$ and construct a boundary piece $Q_i$ as follows; see Figure~\ref{fig:boundary-piece-convex-hull}.
Consider the set
\[
\Sigma_i =\left( \left(C_i\cap P\right) \setminus \bigcup_{j=1}^{i-1} Q_j \right) \cup I_i.
\]
One of the connected components $R$ of $\Sigma_i$ includes the interval $I_i$.
We define our piece as the closure of this region, i.e., $Q_i=\overline R$.
In other words, we create the piece $Q_i$ by ``blowing up'' the interval $I_i$ towards the interior of $P$ and stop when we hit the boundary of $P$, a previously constructed piece, or the boundary of $C_i$.
Clearly, since the interval $I_i$ is of bounded size, so is $Q_i$.
Note that $Q_i$ is a weakly simple but not necessarily a simple polygon:
For instance, $Q_i$ can be contained in the boundary $\partial P$, which happens when $I_i$ is a concave interval on $\partial P$.

\begin{figure}
\centering
\includegraphics[page=14]{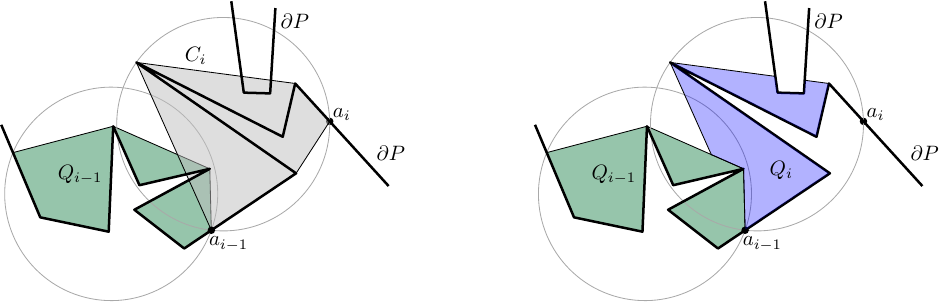}
\caption{Constructing the piece $Q_i$ in a disk partition from the convex hull $C_i$.
}
\label{fig:boundary-piece-convex-hull}
\end{figure}

This construction has the consequence that all intervals on the boundaries of the pieces that are neither shared with $\partial P$ nor with another piece, are contained in the boundaries of the at most $2\,\opt$ convex polygons $\mathcal C$ of bounded size, by~\Cref{lemma:boundarybound}.
As we will see in~\Cref{sec:DSSDinterior}, this makes it possible to give a bound on the number of interior pieces we are going to construct.

To compute the pieces, we consider the overlay of $P$ with $\mathcal C$.
We then use a simple flood fill algorithm from each interval $I_i$, marking the faces as belonging to the piece $Q_i$ until we reach the boundary of $P$, the boundary of a previously constructed piece or the boundary $\partial C_i$ of the convex hull of $I_i$.

\begin{figure}
\centering
\includegraphics[page=4]{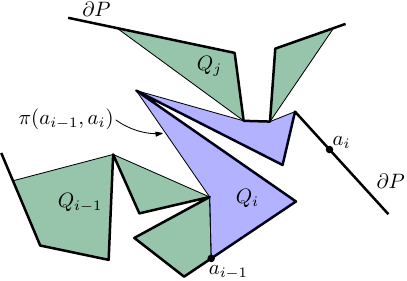}
\caption{Constructing the piece $Q_i$ in a geodesic diameter partition.
}
\label{fig:boundary-piece-geodesic}
\end{figure}

For the geodesic diameter and the perimeter variant, we cannot use the same approach of ``blowing up'' the pieces until we reach the convex hull, since this can cause the geodesic diameter or perimeter to become arbitrarily large.
Instead, we define the boundary piece $Q_i$ as the polygon enclosed by the union of the interval $I_i$ and the shortest path $\pi(a_{i-1},a_i)$ in $P$ between the endpoints of $I_i$; see \Cref{fig:boundary-piece-geodesic}.
We have chosen the interval $I_i$ specifically such that the resulting piece $Q_i$ is of bounded size.
Note that using this construction, our resulting pieces $\mathcal Q_\partial$ are indeed pairwise interior-disjoint, since two shortest paths connecting non-overlapping intervals on the boundary of a simple polygon cannot cross.

We say that an interval $I_i$ and the corresponding boundary piece $Q_i$ are \emph{trivial} whenever $I_i$ does not contain a corner of $P$.
In that case, $I_i$ is just a line segment and $I_i=Q_i$.
In the case of square, disk and straight diameter partitions, we likewise say that the convex hull $C_i\in\mathcal C$ is \emph{trivial} if the corresponding interval $I_i$ is trivial.

\begin{lemma}\label{lemma:boundaryconstruction}
The boundary pieces can be constructed within the time complexities stated in \Cref{thm:main-2}.
\end{lemma}

\begin{proof}
Let us first consider square, disk, and straight diameter partitions.
Recall that when constructing a new piece, we are given a point $a_{i-1}$ on $\partial P$ and want to find the point $a_i$ such that the interval $I_i=\partial P[a_{i-1},a_i]$ is the maximum interval starting at $a_{i-1}$ that can be covered by a single piece.
If $I_i$ is trivial, we use $O(1)$ time to find $a_i$.
We thus use $O(\opt)$ time on intervals of this type in total.
Otherwise, we use $O(s)$ time for aligned square partitions or $O(s\log s)$ time for rotated square, disk and straight diameter partitions, where $s$ is the number of corners of $P$ that were not already covered by other pieces.
In total, we therefore use $O(n+\opt)$ and $O(n\log n+\opt)$ time, respectively, to find the intervals.

We now consider the time used to blow up the intervals.
Here we only have work to do for the non-trivial intervals.
Let $\mathcal C'\subset\mathcal C$ denote the set of non-trivial polygons.
We need to compute the overlay $\mathcal L$ of $P$ and $\mathcal C'$, and we want to show that $\mathcal L$ has complexity $O(n^2)$.
Constructing the actual pieces using a flood fill algorithm then also takes $O(n^2)$ time.

Since each corner of $P$ is in at most two non-trivial intervals in $\mathcal I$, the cardinality of $\mathcal C'$ is at most $2n$.
Furthermore, each polygon $C_i$ has at most two corners that are not corners of $P$, namely the endpoints of the interval $I_i$.
We get that the total number of corners of polygons in $\mathcal C'$ is $O(n)$.
We can then use a standard overlay algorithm to construct $\mathcal L$ in $O(n^2)$ time~\cite[Ch.~2]{de1997computational}.

Let us now consider the time used to find the intervals $I_i$ for geodesic diameter and perimeter partitions.
We can again handle each trivial interval in $O(1)$ time, so all trivial intervals take $O(\opt)$ time in total.
Each non-trivial interval takes $O(n)$ time for perimeter partitions and $O(n\log n)$ time for geodesic diameter partitions, so the total time for all non-trivial intervals is $O(n^2)$ and $O(n^2 \log n)$, respectively.
We construct the actual pieces while finding the intervals, so the total time used is $O(n^2\log n+\opt)$ and $O(n^2+\opt)$, respectively.
\end{proof}

\section{Computing the interior pieces}\label{sec:interiorpieces}
After constructing the boundary pieces as described in~\Cref{sec:boundary-pieces}, we place a regular square grid on top of $P$, as shown in \Cref{fig:interiorpiecetypes}.
We choose the size of the grid according to the type of partition we are constructing, which is to be specified in later sections for the particular types of partitions.
We define $\mathcal S$ to be the set of the squares that intersect~$P$.
We create interior pieces from the regions in the squares that are not contained in boundary pieces.

\begin{figure}
\centering
\includegraphics[page=8]{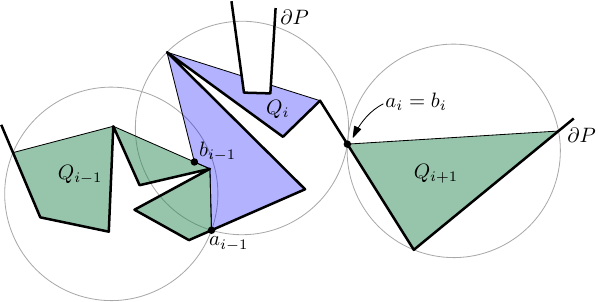}
\caption{Three consecutive boundary pieces $Q_{i-1},Q_i,Q_{i+1}$ of a disk partition.
}
\label{fig:pointsab}
\end{figure}

In order to bound the number of interior pieces, we need to introduce the notion of a \emph{free interval}, as follows.
Let $\CQ_\partial=\{Q_1,\ldots,Q_k\}$ be the boundary pieces in the counterclockwise order in which they were constructed by the greedy algorithm in \Cref{sec:boundary-pieces}.
For each $i\in\{1,\ldots,k\}$, the boundary of $Q_i$ and $Q_{i+1}$ share an interval $\partial Q_i[a_i,b_i]=\partial Q_{i+1}[b_i,a_i]$, where indices are taken modulo $k$; see~\Cref{fig:pointsab}.
Here, $a_i$ is a point on $\partial P$, while $b_i$ need not be.
Each piece $Q_i$ was chosen in a greedy way so as to maximize the interval $\partial Q_i[a_{i-1},a_i]\subset \partial P$ that $\partial Q_i$ covers of $\partial P$.
In other words, the exterior of $P$ is to the right of this interval (when traversing $\partial Q_i$ in counterclockwise direction), while the interval $\partial Q_i(b_i,b_{i-1})$ has the interior of $P$ to the right (except where it occasionally meets the interval $\partial P(a_i,a_{i-1})$ that constitute the rest of $\partial P$ that the piece $Q_i$ could not cover).
We define a \emph{free interval} of $Q_i$ to be a maximal interval in $\partial Q_i[b_i,b_{i-1}]$ that is interior-disjoint from all other boundary pieces.

\begin{lemma}\label{lemma:ibi}
There are at most $6\,\opt-9$ free intervals.
\end{lemma}

\begin{figure}
\centering
\includegraphics[page=9]{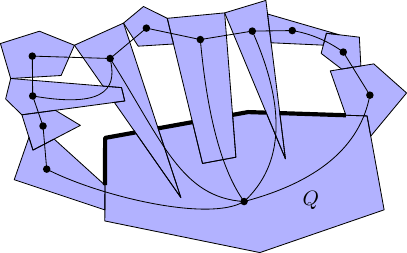}
\caption{The graph $G$ from the proof of \Cref{lemma:ibi}.
The polygons are the boundary pieces.
There are four free intervals on the boundary of the piece $Q$, shown in fat.
The figure is for illustration only; the pieces produced by our algorithms would not look like these.
}
\label{fig:graphG}
\end{figure}

\begin{proof}
Consider the undirected graph $G$ with vertices $\CQ_\partial$ where we add an edge between two pieces $Q_i$ and $Q_j$ if $\partial Q_i\cap \partial Q_j\neq\emptyset$; see \Cref{fig:graphG}.
Note that $G$ is an outerplanar graph where the outer face is bounded by a Hamiltonian cycle $C$ consisting of the edges $Q_iQ_{i+1}$ corresponding to the boundary pieces in cyclic order.
If there were no other edges in $G$, there would be exactly $\vert\CQ_\partial\vert$ free intervals, namely the intervals $\partial Q_i[b_i,b_{i-1}]$.
Each extra edge of $G$ between pieces $Q_i$ and $Q_j$ can cause the creation of two extra free intervals, one on each of $\partial Q_i$ and $\partial Q_j$.
Note that $G$ has at most $\vert \CQ_\partial\vert-3$ of these extra edges, since that number of edges will cause the cycle $C$ to be triangulated.
Hence, there are at most $\vert \CQ_\partial\vert+2(\vert \CQ_\partial\vert-3)=3\vert \CQ_\partial\vert-6$ free intervals.
Recall that $\vert \CQ_\partial\vert\leq 2\,\opt-1$ due to~\Cref{lemma:boundarybound}.
We therefore get that there are at most $6\,\opt-9$ free intervals.
\end{proof}

\subsection{Disk, square and straight diameter partitions}\label{sec:DSSDinterior}
We choose the edge length of the grid squares $\mathcal S$ to be maximum so that each square satisfies the size constraint, i.e., the edge length of each square should be $\sqrt 2$, $1/\sqrt 2$ and $1$ for disk, straight diameter and square partitions, respectively.

Consider a square $S\in\mathcal S$.
We define a \emph{field} to be the closure of a connected component of $(S\cap P)\setminus\bigcup\mathcal Q_\partial$.
We use each field as an interior piece.
That is, if no boundary piece overlaps the square $S$, then $S$ itself is used as a piece. Otherwise, $S$ may be broken into several interior pieces by the intersecting boundary pieces.
For disk, square and straight diameter partitions, we can now bound the number of interior pieces as stated by the following lemma.

\begin{lemma}\label{lemma:DSDPinterior}
The number of interior pieces are bounded as follows.
\begin{itemize}
\item Aligned square partitions: $11\,\opt$.

\item Rotated square partitions: $19\,\opt$.

\item Disk and straight diameter partitions: $(18+\pi/2)\,\opt$.
\end{itemize}
\end{lemma}

\begin{proof}
We say that an interior piece $Q$ is \emph{complete} if it is a full square in $\mathcal S$, and otherwise it is \emph{incomplete}; see \Cref{fig:interiorpiecetypes}.
We first bound the number of complete pieces as follows.
For disk partitions, the complete pieces have area $2$, so we get that the number of complete pieces is at most $\area P/2$.
For straight diameter partitions, each complete piece has area $1/2$, so there are at most $2\area P$ complete pieces.
For square partitions, each complete piece has area $1$, so there are at most $\area P$ complete pieces.

\begin{figure}
\centering
\includegraphics[page=10]{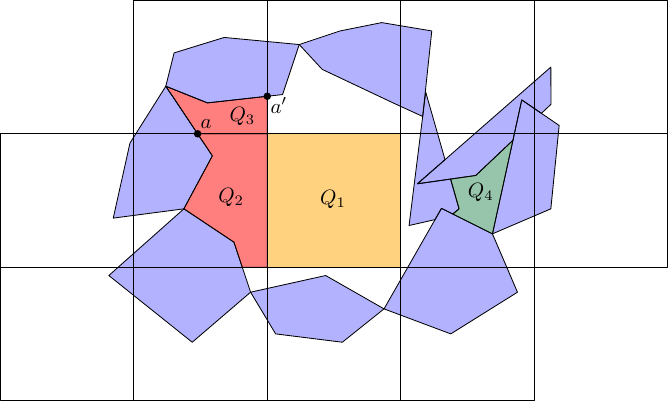}
\caption{The figure shows the three types of interior pieces.
The blue polygons are the boundary pieces.
The square $Q_1$ is a complete interior piece, $Q_2$ and $Q_3$ are edge pieces, and $Q_4$ is a chip piece.
We use the intersection point $a$ to account for $Q_2$ while $a'$ accounts for $Q_3$
The figure is for illustration only; the pieces produced by our algorithms would not look like these (see \Cref{fig:SquarePartition} for an authentic example).
}
\label{fig:interiorpiecetypes}
\end{figure}

In disk, straight diameter and square partitions, we have $\opt\geq \area P/\pi$, $\opt\geq 4\area P/\pi$ and $\opt\geq \area P$, respectively.
We therefore get that the number of complete pieces is at most $\pi\,\opt/2$ for disk and straight diameter partitions and at most $\opt$ for both kinds of square partitions.

We now turn our attention to the incomplete pieces.
We distinguish between two types of incomplete pieces, as follows.
If an interval of the boundary of a square $S\in\mathcal S$ appears on the boundary of an incomplete piece $Q$, we say that $Q$ is an \emph{edge piece}, and otherwise, it is a \emph{chip piece}.
For the ease of presentation, we make the following general position assumptions:
(i) no polygon in $\mathcal C$ has a left-, right-, top- or bottom-most point on the boundary of a square in $\mathcal S$, and (ii) no square in $\mathcal S$ has a corner on the boundary of a polygon in $\mathcal C$.
This can be obtained either by symbolic perturbation~\cite{edelsbrunner1990simulation} or by choosing the grid defining the squares $\mathcal S$ carefully.
The lemma is also true without these assumptions, but with the assumptions, we can skip some tedious special cases.

Consider an edge piece $Q$ in a square $S$ and let $I$ be a maximal interval in $\partial Q\cap \partial S$.
Let $a$ be the end of $I$ in counterclockwise order around $\partial Q$.
Then $a$ is an intersection point in $\partial C\cap \partial S$ for a convex polygon $C\in\mathcal C$.
In this way, we can account for each edge piece by such an endpoint.
Note that although the point $a$ is also on the boundary of an edge piece $Q'$ in a neighbouring square $S'$ of $S$, we will not use $a$ to account for $Q'$, since $a$ will be the beginning and not the end of a maximal interval in $\partial Q'\cap \partial S'$.

Let us first consider rotated square, disk and straight diameter partitions.
Here, a polygon $C\in\mathcal C$ can intersect at most two vertical lines and at most two horizontal lines in the grid that defines the squares $\mathcal S$, so there are at most eight intersection points; see \Cref{fig:intersectionpoints}.
Using that $\vert \mathcal Q_\partial\vert< 2\,\opt$ by \Cref{lemma:boundarybound}, we get that there are less than $16\,\opt$ intersection points between the boundaries of polygons in $\mathcal C$ and the boundaries of the squares $\mathcal S$.
We conclude that there are less than $16\,\opt$ edge pieces in rotated square, disk and straight diameter partitions.

\begin{figure}
\centering
\includegraphics[page=11]{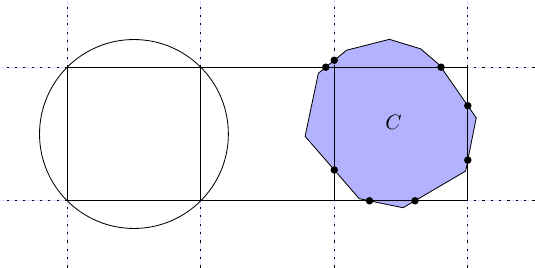}
\caption{The figure shows the grid we use for disk partitions.
The convex polygon $C$ can be contained in a unit disk, and its boundary has eight intersection points with the lines defining the grid, which is the maximum.
}
\label{fig:intersectionpoints}
\end{figure}

In aligned square partitions, each square $S\in\mathcal S$ is an axis-aligned unit square, and each polygon $C\in\mathcal C$ is contained in an axis-aligned unit square.
Therefore, $C$ intersects at most one vertical and one horizontal line of the square grid, so there are at most four intersection points.
We therefore get less than $8\,\opt$ edge pieces in this case.

The boundary of each chip piece $Q$ is a union of three or more free intervals, and each of these intervals is a concave interval on $\partial Q$.
By \Cref{lemma:ibi} we then get that there are less than $6\,\opt/3=2\,\opt$ chip pieces for all types of partitions.

In total, we get less than $(18+\pi/2)\,\opt$ interior pieces for disk and straight diameter partitions, less than $11\,\opt$ interior pieces for aligned square partitions, and less than $19\,\opt$ interior pieces for rotated square partitions.
\end{proof}

We can now conclude our work on these types of partitions by giving the following proof.

\begin{proof}[Proof of \Cref{thm:main-2} for square, disk and straight diameter partitions.]
The approximation ratios follow from \Cref{lemma:boundarybound,lemma:DSDPinterior}, so here we focus on describing algorithms with the claimed running times for the estimation and construction problems.
For the estimation problem, we first compute the cardinality $x_\partial$ of our boundary partition $\CQ_\partial$ using the stated running times for estimation, where we ``jump over'' the trivial intervals to avoid an additive term of $\opt$ in the time complexity.
That is, suppose that our starting point $a_{i-1}$ is on an edge $e$ of $P$ and let $p$ be the end of $e$ so that the segment $a_{i-1}p\subset e$ remains to be covered by boundary pieces.
By first computing the length of $a_{i-1}p$, we can in $O(1)$ time compute how many trivial intervals our algorithm will construct on $e$ until we get to a non-trivial one, which is the first containing $p$.
For aligned square partitions, we know after $x_\partial$ has been computed that our constructed partition of all of $P$ would have at most $4x_\partial$ edge pieces, $x_\partial$ chip pieces and $\lfloor\area(P)\rfloor$ complete pieces.
We also know that the sum of these numbers is at most $13\opt$.
We hence return $6x_\partial+\lfloor\area(P)\rfloor$.

For rotated square partitions, we similarly return $10x_\partial+\lfloor\area(P)\rfloor$, because we have at most $8x_\partial$ edge pieces in this case.
For disk and straight diameter partitions, we return $10x_\partial+\lfloor\area(P)/2\rfloor$ and $10x_\partial+\lfloor 2\area(P)\rfloor$, respectively, adjusting for the size of the grids we use.
The running times are explained in the proof of \Cref{lemma:boundaryconstruction}.

For the construction of the interior pieces, we compute the overlay $\mathcal L$ of all free intervals and the squares $\mathcal S$.
Note that there are $O(\opt)$ squares since each piece in an optimal partition can intersect at most $O(1)$ squares in $\mathcal S$.
The total complexity of the free intervals (not counting intersections with the squares $\mathcal S$) is $O(n+\opt)$, since each corner is a corner of $P$ or an endpoint of an interval, and there are $O(\opt)$ intervals by~\Cref{lemma:ibi}.
The free intervals of a single convex polygon $C_i$ have $O(1)$ intersection points with the boundaries of the squares $\mathcal S$.
Therefore, the complexity of $\mathcal L$ is $O(n+\opt)$.

Using a standard overlay algorithm~\cite[Ch.~2]{de1997computational}, we would get a running time of $O((n+\opt)\log(n+\opt))$, but in this special case where $\mathcal S$ are just squares from a regular grid, we can use an improved version with running time $O((n +\opt)\log n)$:
The idea is that we do not need to handle the edges of $\mathcal S$ in the same way as the edges of the free intervals.
We use a horizontal sweepline $\ell$ starting above $P$ and going down.
We define a \emph{free edge} to be an edge of a free interval.
There are four types of events, defined by the $y$-coordinates of the following objects: (i) the top endpoint of a free edge, (ii) the bottom endpoint of a free edge, (iii) a horizontal grid line $h$, and (iv) the intersection point between a free edge and a vertical grid line.
We store the free edges currently intersected by $\ell$ in a status data structure $\mathcal T$, which is a balanced binary search tree sorted by the $x$-coordinates of the intersections between the sweepline $\ell$ and the edges.
For each vertical grid line $k$, we also store a pointer $a_k$ to the vertex of $\mathcal L$ on $k$ that is immediately above the sweepline $\ell$ (this may be a gridpoint or an intersection point between a free edge and $k$).
For each free edge $e$ in $\mathcal T$, we likewise store a pointer $b_e$ to the vertex of $\mathcal L$ on $e$ that is immediately above $\ell$, which may be the top endpoint of $e$ or an intersection between $e$ and a grid line.
The events are handled as follows.
\begin{itemize}
\item[(i)] We add the top endpoint $q$ of the new edge $e$ to $\mathcal L$ and set $b_e$ to point on $q$.
We add $e$ to $\mathcal T$.
We also check if $e$ intersects one of the neighbouring vertical grid lines, and if so, we add the corresponding event of type (iv) to the event queue.

\item[(ii)] We add the bottom endpoint $q$ of the edge $e$ to $\mathcal L$ and construct an edge from $q$ to the preceding vertex $b_e$.
We then remove $e$ from $\mathcal T$.

\item[(iii)]
We run through the entire tree $\mathcal T$ in order of $x$-coordinates and construct the bottom part of the row of squares in $\mathcal S$ bounded from below by the horizontal line $h$.
This involves finding the intersections between the free edges in $\mathcal T$ and $h$ and constructing the segments from these intersection points to the preceding vertices on the edges above $\ell$.
We update all the pointers $a_k$ to point to the new grid points on $h$ and we update the pointers $b_e$ to point at the intersection points between the free edges and $h$.

\item[(iv)]
When we reach the intersection point $q$ of a free edge $e$ with a vertical grid line $k$, we add $q$ to $\mathcal L$ and construct the edges from $q$ to the preceding vertices $a_k$ and $b_e$.
We update $a_k$ and $b_e$ to point to $q$ and also check if $e$ intersects a neighbour of $k$ below $\ell$, in which case we add the corresponding event of type (iv) to the event queue.
\end{itemize}

Events of types (i), (ii) and (iv) take $O(\log n)$ time to handle, and since there are $O(n+\opt)$ of them, they take $O((n +\opt)\log n)$ time in total.
In events of type (iii), we use constant time for each edge we construct of the overlay $\mathcal L$, so since $\mathcal L$ has complexity $O(n+\opt)$, we use that much time on events of type (iii) in total.
After computing $\mathcal L$, we just need to detect which faces are fields.
This can be done by a flood fill algorithm, starting from each free interval $I\subset\partial C_i$ and traversing the faces towards the exterior of the polygon $C_i$.\qedhere
\end{proof}

\subsection{Geodesic diameter and perimeter partitions}\label{sec:interiorgeodesic}

For geodesic diameter and perimeter partitions, we again use a square grid $\mathcal S$, where the edge length of each square is a constant $\gamma >0$ to be chosen later.
Recall that a \emph{field} in a square $S\in\mathcal S$ is the closure of a connected component of $(S\cap P)\setminus\bigcup\mathcal Q_\partial$, and that $\mathcal I$ is the intervals on $\partial P$ that are used to define the boundary pieces $\mathcal Q_\partial$.
We say that a field $F$ is \emph{trivial} if no non-trivial interval from $\mathcal I$ appears on the boundary of $F$; see \Cref{fig:fields}.
Trivial intervals can appear on the boundary of a trivial field $F$, but these intervals must be contained in edges of $P$ that pass all the way through the square $S$ containing $F$.
In particular, a trivial field is a convex polygon contained in $S$.
Hence, a trivial field satisfies the size constraint, and we will therefore use it directly as a piece.

\begin{figure}
\centering
\includegraphics[page=6]{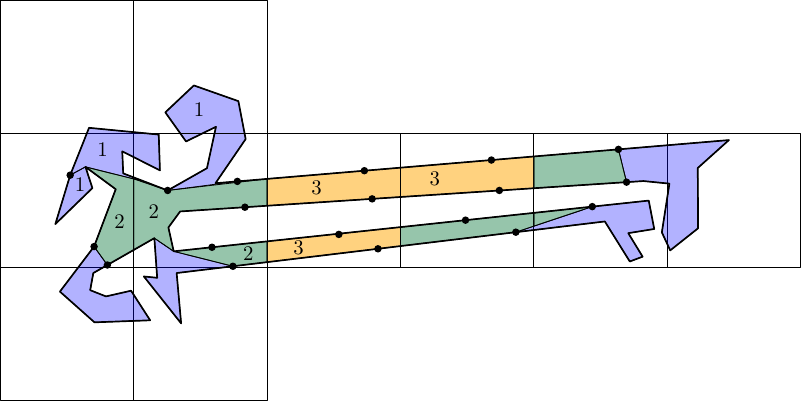}
\caption{A polygon $P$ with a perimeter partition.
The blue pieces are boundary pieces.
The dots mark the endpoints of the intervals $\mathcal I$ covered by the boundary pieces.
The green regions are the non-trivial fields and yellow regions are the trivial fields.
The figure is for illustration only, and our algorithm would use a grid of much smaller squares.
Color codes: 1 blue, 2 green, 3 yellow.
}
\label{fig:fields}
\end{figure}

Consider a non-trivial field $F$ in a square $S$; see \Cref{fig:perimetersemiboundary} (left).
We need to be careful when defining the interior pieces for geodesic diameter and perimeter partitions:
We cannot directly use $F$ as a piece (as we did in \Cref{sec:DSSDinterior}), since $F$ can have arbitrarily large geodesic diameter or perimeter.
Instead, we split $F$ into \emph{subfields}, as described in the following.

We first split the free intervals into subintervals, that we call \emph{fragments}.
Each fragment can have length at most $\delta$ for some constant $\delta>0$ depending on the type of partition.
We split each free interval into a minimum number of fragments by traversing the interval and repeatedly splitting whenever we have traversed a prefix of length $\delta$.

Let us observe that the free intervals, and hence also the fragments, are convex intervals on the boundaries of the boundary pieces (just as in disk, square and straight diameter partitions):
Recall that each boundary piece $Q_i$ is bounded by a maximal interval $\partial P[a_{i-1},a_i]$ and the shortest path $\pi=\pi_P(a_{i-1},a_i)$ in $P$ between the endpoints $a_{i-1}$ and $a_i$.
Consider an interior corner $v$ of $\pi$ that is concave with respect to $Q_i$.
Such a corner $v$ is on the shared boundaries between $Q_i$ and another boundary piece $Q_j$, so $v$ is not in the interior of a free interval of $Q_i$, and it follows that all free intervals are convex intervals.
Let us remark that despite this fact, the situation is still more complicated than for disk, square and straight diameter partitions, as a single free interval for geodesic diameter and perimeter partitions can be a spiral that makes an arbitrary number of revolutions and thus, for instance, enters and leaves the same square in $\mathcal S$ an arbitrary number of times.

\begin{figure}
\centering
\includegraphics[page=5]{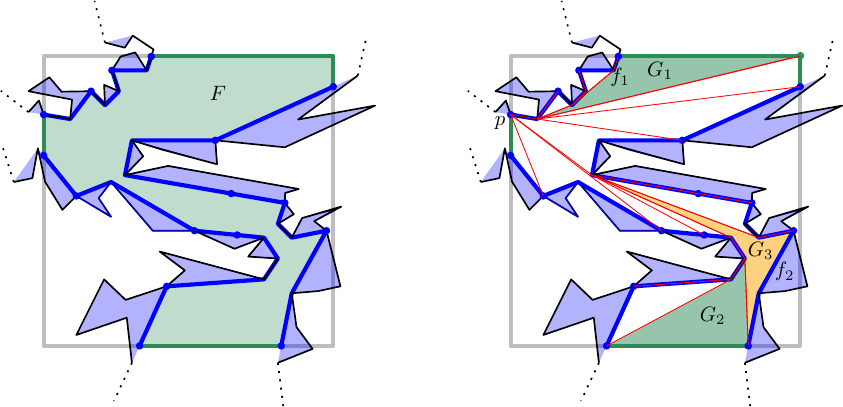}
\caption{The figure shows how we split a field $F$ in a square $S$ into subfields and assign them to the fragments.
The boundary pieces are blue, the fat blue curves are the fragments, and the blue points are the endpoints of the fragments (restricted to $S$).
The boundary of $P$ is sketched in black (not to scale; it should be much more complicated to result in this many boundary pieces in a single square).
To the right is shown how the field $F$ is split into subfields by shortest paths to the point $p$.
The subfields $G_1$ and $G_2$ are edge subfields, and $G_3$ is a fragment subfield.
Here, $G_1$ is assigned to the fragment $f_1$ and $G_2$ and $G_3$ are assigned to $f_2$.}
\label{fig:perimetersemiboundary}
\end{figure}

To define the subfields in a field $F$, we choose an arbitrary fixed point $p$ in $F$ (for instance, a corner of $F$).
For each point $c$ of $\partial F$ that is either a corner of $\partial S$ or a shared endpoint of two fragments bounding $F$, we cut $F$ along the geodesic shortest path $\pi_F(c,p)$ in $F$ from $c$ to $p$.
These paths partition $F$ into \emph{subfields}.
There exist two types of subfields:
A subfield $G$ is an \emph{edge subfield}
if a segment $s$ on an edge of $S$ appears on the boundary of $G$.
Otherwise, $G$ is a \emph{fragment subfield}, and here an interval $f'$ on a fragment $f$ bounds $G$ (the full fragment $f$ may also intersect other squares in $\mathcal S$, which is why only an interval on $f$ appears on the boundary of $G$ in general).

Consider an edge subfield $G$ and let $a$ and $b$ be the endpoints of the segment $s\subset\partial S$; see \Cref{fig:fragmentacrosssquares} (left).
Let $\pi_a=\pi_F(a,p)$ and $\pi_b=\pi_F(b,p)$ be the shortest paths in $F$ from $a$ and $b$ to $p$, respectively, and let $q$ be the first common point of $\pi_a$ and $\pi_b$.
Then the boundary of $G$ consists of $s$ and the prefixes of $\pi_a$ and $\pi_b$ until the common point $q$.
An analogous description holds for a fragment subfield $G$, where we define $a$ and $b$ as the endpoints of the interval $f'$ of a fragment bounding $G$.

We assign a fragment subfield $G$ to the fragment $f$ bounding $G$; see \Cref{fig:fragmentacrosssquares} (right).
For an edge subfield $G$ bounded by a segment $s\subset\partial S$, we consider the next fragment $f$ following $s$ on $\partial F$ in counterclockwise order and then assign $G$ to $f$.
Note that since a fragment can intersect more than one square, a piece is not necessarily contained in a single square, but it will be contained in the union of at most a small constant number of squares, since we will define $\delta$ to be smaller or not much larger than $\gamma$.

Using appropriate values of $\gamma$ and $\delta$, we can now bound the number of pieces as stated in \Cref{lemma:GDinterior,lemma:perimeterbound}.
A key insight in the proofs is that since the boundary of each subfield consists of three concave intervals in a square of size $\gamma\times\gamma$, the geodesic diameter and perimeter of each subfield is $O(\gamma)$.

\begin{figure}
\centering
\includegraphics[page=26]{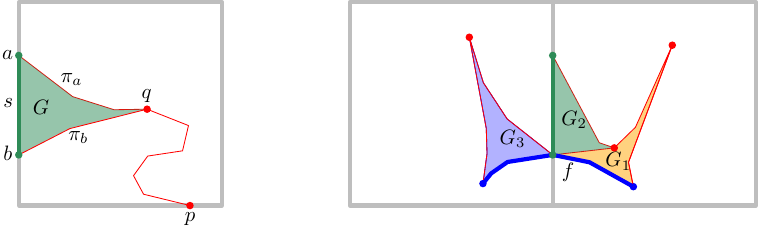}
\caption{Left:
An edge subfield bounded by a segment $s\subset\partial S$ and the prefixes of the shortest paths $\pi_a$ and $\pi_b$ until their first common point $q$.
Right: Three subfields assigned to one fragment $f$, so the union $G_1\cup G_2\cup G_3$ will constitute one piece.
Here, $G_1$ and $G_3$ are fragment subfields and $G_2$ is an edge subfield.}
\label{fig:fragmentacrosssquares}
\end{figure}

\begin{lemma}\label{lemma:GDinterior}
Choosing $\gamma=0.127$ and $\delta=0.133$, the method described above results in less than $70\,\opt$ interior pieces of geodesic diameter at most $1$.
\end{lemma}

\begin{proof}
We first consider the complete interior pieces, i.e., the squares from $\mathcal S$ that are also used as pieces.
Our algorithm creates at most $\area P/\gamma^2$ such pieces.
A disk of diameter $1$ has the largest area of all shapes with geodesic diameter at most $1$, so we have $\opt\geq 4\area P/\pi$, and the number of complete pieces is thus at most $\area P/\gamma^2\leq \frac{\pi\,\opt}{4\gamma^2}$.

We make the analysis as if we were splitting all remaining fields into subfields, although the algorithm will actually only split the non-trivial fields; this just gives an upper bound on the resulting number of pieces.
Note that each edge or fragment subfield subfield $G$ is a so-called \emph{pseudo-triangle}, i.e., $G$ is bounded by three concave intervals (for edge subfields, one of these concave intervals is just a segment on the boundary of $S$).
Hence, every subfield $G$ is contained in a triangle $T$, whose corners $a,b,c$ are the endpoints of the three concave intervals.
The concave interval from $a$ to $b$ has length less than $\Vert ac\Vert+\Vert cb\Vert$, and the length of the other two intervals can be bounded in a similar way.
We obtain that the geodesic diameter of $G$ is less than the perimeter of $T$.
It is easy to verify that a triangle in the square $S$ with maximum perimeter is spanned by three corners of $S$, so the perimeter of $T$ is at most $(2+\sqrt 2)\gamma$, which is also a bound on the geodesic diameter of each subfield $G$.

Consider now two points $u$ and $v$ in the piece assigned to a fragment $f$.
Then $u$ and $v$ are contained in subfields $G_u$ and $G_v$, respectively, and we get a path from $u$ to $v$ by first walking in $G_u$ to $f$ and then along $f$ to $G_v$, from which we go to $v$.
Since the geodesic diameter of each subfield is at most $(2+\sqrt 2)\gamma$ and $f$ has length $\delta$, this path has length at most $2(2+\sqrt 2)\gamma+\delta$.
To bound the geodesic diameter, we therefore require that $2(2+\sqrt 2)\gamma+\delta\leq 1$.
Recall that the fragments are contained in $2\,\opt$ shortest paths of length at most $1$ by \Cref{lemma:boundarybound}, so we split the free intervals at most $2\,\opt/\delta$ times in order to create the fragments.
We start with at most $6\,\opt$ free intervals by \Cref{lemma:ibi}.
We conclude that we create at most $6\,\opt + 2\,\opt/\delta$ fragments.

In total, we create at most $(\frac{\pi}{4\gamma^2} + 6 + 2/\delta)\,\opt$ interior pieces.
Under the restriction that $2(2+\sqrt 2)\gamma +\delta\leq 1$, this is minimized for $\gamma\approx 0.127$ and $\delta=1-2(2+\sqrt 2)\gamma\approx 0.133$, in which case the number of piece is bounded by a bit less than $70\,\opt$.
\end{proof}

\begin{lemma}\label{lemma:perimeterbound}
Choosing $\gamma=0.00490$ and $\delta=0.00243$, the described method results in less than $3726\,\opt$ interior pieces of perimeter at most $1$.
\end{lemma}

\begin{figure}
\centering
\includegraphics[page=12]{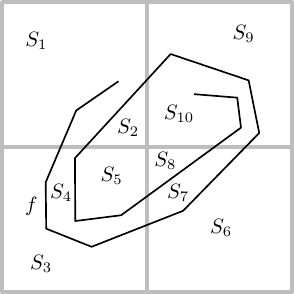}
\caption{The figure shows how a fragment $f$ can partition four squares into regions $S_1,\ldots, S_{10}$ when constructing a perimeter partition.
The figure is not to scale ($f$ is too long here).}
\label{fig:spiralfragment}
\end{figure}

\begin{proof}
Note that a disk of radius $\frac 1{2\pi}$ has the largest area of any shape of perimeter $1$.
The complete interior pieces in our partition have area $\gamma^2$, so there are at most $\frac \opt{4\gamma^2\pi}$ complete pieces.

We now turn our attention to the number of pieces we create from incomplete fields.
As in the proof of \Cref{lemma:GDinterior}, assume that we split all fields that are not full squares into subfields, which gives an upper bound on the number of pieces.
We require that $\delta\leq \gamma$ so that a fragment can intersect at most four different squares in $\mathcal S$.
Consider a fragment $f$ and let $\mathcal S'\subset\mathcal S$ be the set of up to four squares intersected by $f$.
The fragment $f$ partitions the squares $\mathcal S'$ into regions $S_1,\ldots,S_m$; see \Cref{fig:spiralfragment}.
Each region $S_i$ is bounded by some intervals on $f$ and some intervals on the boundary of a square $S\in\mathcal S'$.
The regions $S_i$ are, in general, not contained in $P$, but we can use them to bound the perimeter of the subfield that constitute the pieces of our partition.
Note that for a square $S\in\mathcal S$, each interval $f'$ in $f\cap S$ appears twice on the boundaries of regions $S_i$ in $S$: either there are different regions on the two sides of $f'$, or the two sides of $f'$ contribute to the perimeter of the same region.
We therefore get
\[
\sum_{i=1}^m \peri S_i \leq 2\delta+16\gamma.
\]
Each subfield from $S\in\mathcal S'$ that we assign to $f$ is contained in one of these regions $S_i$.
As in the proof of \Cref{lemma:GDinterior}, we note that a subfield $G$ in $S_i$ is bounded by three concave intervals and thus contained in a triangle $T$.
It follows that $\peri G\leq 2\peri T\leq 2\peri S_i$, where the latter inequality follows as the corners of $T$ are in $S_i$.
Note that at most six subfields in $S_i$ are assigned to $f$: there can be one fragment subfield and at most five edge subfields (for one edge of $S$, there can be two edge subfields and for the others at most one).
The total perimeter of subfields in $S_i$ assigned to $f$ is therefore  at most $12\peri S_i$.
We then get that the total perimeter of subfields that are assigned to $f$, and thus the perimeter of the resulting interior piece, is at most
\[
\sum_{i=1}^m 12\peri S_i \leq 12(2\delta+16\gamma)=24\delta+192\gamma.
\]

To get pieces of bounded perimeter, we hence require that $24\delta+192\gamma\leq 1$.
Consider an interval $I_i=\partial P[a_{i-1},a_i]$ and the shortest path $\pi=\pi(a_{i-1},a_i)$.
We note that since $Q_i$ has perimeter at most $1$, the length of $\pi$ is at most $1/2$.
Similarly as in the proof of \Cref{lemma:GDinterior}, we can then bound the number of fragments by $(6+ 1/\delta)\,\opt$, and the total number of interior pieces (complete and incomplete) is at most $(\frac 1{4\gamma^2\pi} + 6 + 1/\delta)\,\opt$.
We obtain the minimum under the constraint that $24\delta+192\gamma\leq 1$ when $\gamma\approx 0.00490$ and $\delta=1/24-8\gamma\approx 0.00243$, in which case we get less than $3726\,\opt$ interior pieces.
\end{proof}

We can now conclude our work on geodesic diameter and perimeter partitions by giving the following proof.

\begin{proof}[Proof of \Cref{thm:main-2} for geodesic diameter and perimeter partitions.]
The approximation ratios follow from \Cref{lemma:boundarybound,lemma:GDinterior,lemma:perimeterbound} and the analysis of the running times of solving the estimation problems is analogous to that for the other types of partitions described in \Cref{sec:DSSDinterior}, so here we focus on the time required to solve the construction problems.

As in the proof for the other types of partitions, we construct the interior pieces by first computing the overlay $\mathcal L$ of the squares $\mathcal S$ and the free intervals, and we can again conclude that $\mathcal L$ has complexity $O(n+\opt)$ and can be constructed in $O((n+\opt)\log n)$ time.
We can then identify the trivial and non-trivial fields in time $O(n+\opt)$ by traversing the overlay $\mathcal L$.
The trivial fields are used as interior pieces, and it remains to partition the non-trivial fields into subfields.

To get a bound on the time needed to split the non-trivial fields, we show that the total complexity of all non-trivial fields is $O(n)$.
We define a free interval $I\subset\partial Q$ to be \emph{non-trivial} if $Q$ is a non-trivial piece.
Note that the corners of the non-trivial free intervals from a boundary piece $Q_i\in\mathcal Q_\partial$ are corners of $P$ except possibly the endpoints $a_{i-1}$ and $a_i$ of the interval $I_i$ defining $Q_i$.
Each corner $c$ of $P$ appears on at most four free intervals, namely one starting and one ending at $c$ and one starting and one ending on another piece on whose boundary $c$ forms a concave corner.
We therefore get that the total complexity of all non-trivial free intervals is $O(n)$.

\begin{figure}
\centering
\includegraphics[page=7]{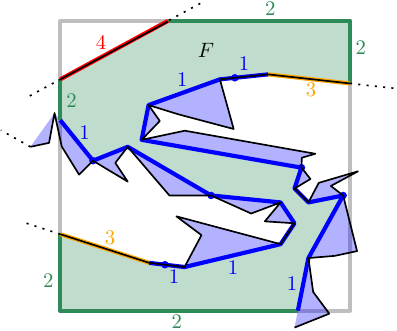}
\caption{A field $F$ with segments of the four types \ref{edgetype:1}--\ref{edgetype:4}.}
\label{fig:segmenttypes}
\end{figure}

Consider a non-trivial field $F$ in a square $S$ and an edge $s\subset\partial F$ of $F$.
We claim that $s$ must have one of the following types; see \Cref{fig:segmenttypes}.
\begin{enumerate}
\item $s$ overlaps with an edge of a non-trivial free interval.\label{edgetype:1}
\item $s\subset \partial S$.\label{edgetype:2}
\item $s$ is contained in a trivial interval in $\mathcal I$ and is incident to an edge of type~\ref{edgetype:1}.\label{edgetype:3}
\item $s$ is contained in a trivial interval in $\mathcal I$ and is incident to two edges of type~\ref{edgetype:2}.\label{edgetype:4}
\end{enumerate}
To see that the list is exhaustive, note that if an edge $s$ of $F$ is contained in a trivial interval, it cannot be incident to another such edge, since their common endpoint would then have to be a corner of $P$, so the intervals would not be trivial.
We now observe that there can be at most $O(1)$ edges of types~\ref{edgetype:2}--\ref{edgetype:4} for each edge of type \ref{edgetype:1}:
Consider a maximal interval $s_1\cup \ldots\cup s_m$ of edges of $F$ of types~\ref{edgetype:2}--\ref{edgetype:4}.
Then $s_1$ and $s_m$ have type~\ref{edgetype:3} while the rest have types \ref{edgetype:2} and \ref{edgetype:4}. 
Each edge of type \ref{edgetype:4} cuts off a corner from $S$, so there can be at most four of them.
We cannot have more than five edges of type \ref{edgetype:2} in a row, because otherwise we would have traversed $\partial S$ completely.
Hence, we have $m=O(1)$, and we can use the edge of type \ref{edgetype:1} after $s_m$ to account for the edges $s_1,\ldots,s_m$.

Each edge of a non-trivial free interval can intersect at most $O(1)$ squares in $\mathcal S$, since the squares have constant edge length $\gamma>0$, so each such edge accounts for $O(1)$ edges of non-trivial fields in total.
Since the non-trivial free intervals have total complexity $O(n)$, we conclude that the total number of edges of type \ref{edgetype:1} is $O(n)$.
Then there are $O(n)$ edges all in all, so the complexity of the non-trivial fields is also $O(n)$.

Since the total length of the non-trivial boundary intervals is $O(n)$ and we split them into fragments of length $\delta>0$, there are likewise $O(n)$ fragment endpoints on the boundaries of the non-trivial fields.
In a non-trivial field $F$ with $n'$ corners, we can compute the shortest paths from the special point $p$ in $O(n')$ time~\cite{DBLP:journals/algorithmica/GuibasHLST87} and then construct the subfields in $F$ in the same time.
We can therefore construct all subfields in $O(n)$ time.
\end{proof}

\section{The area partitioning problem}\label{sec:area}
In this section, we consider the area partitioning problem and consequently prove \cref{thm:main-1}. 
In the area partitioning problem we are given a simple polygon $P$ and $k$ positive real values $a_1, a_2, \ldots, a_k$ such that $\area P = \sum_{i=1}^k a_i$. The goal is to compute a partition of $P$ into exactly $k$ pieces $Q_1, Q_2, \ldots, Q_k$ such that $Q_i$ is a simple polygon and $\area Q_i  = a_i$. 

\begin{remark}
We will assume that $P$ is a simple polygon, but the algorithm can easily be extended to work for polygons with holes. The only additional step is to first ``remove the holes'' of $P$ by connecting the holes to the outer boundary of $P$ via diagonals, ensuring that the final polygon remains connected.
However, the running time becomes $O(n\log n+k)$ instead of $O(n+k)$, since triangulating a polygon with holes takes $O(n\log n)$ time~\cite{de1997computational}.
\end{remark}

\begin{figure}
\centering
\includegraphics[page=15]{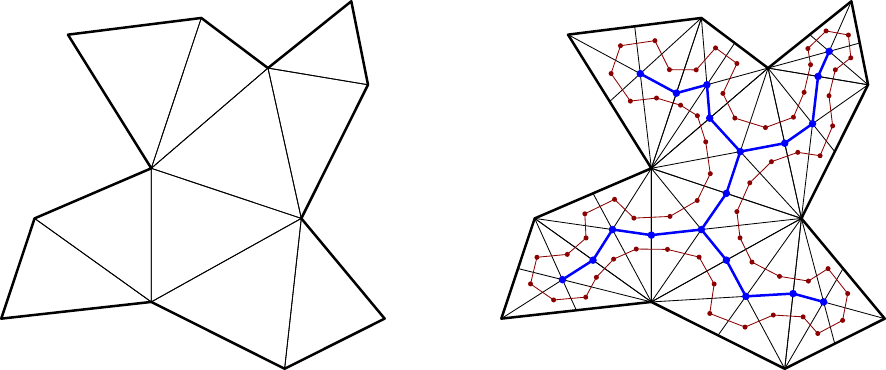}
\caption{Left:
The triangulation $\T$ of the input polygon $P$.
Right: The Steiner triangulation $\T_S$ with the wall shown in blue and the hamiltonian cycle $C$ in red.}
\label{fig:triangulations}
\end{figure}

We first compute an ordinary triangulation $\T$ of $P$, i.e., without using Steiner points.
We then split each triangle $\triangle\in \T$ into six pieces by adding the medians of $\triangle$. The resulting Steiner triangulation is denoted $\T_S$; see~\Cref{fig:triangulations}.
We can compute the original triangulation $\T$ in $O(n)$ time using Chazelle's algorithm~\cite{chazelle1991triangulating}, and then it is easy to compute $\T_S$ in $O(n)$ time as well (we note that $\T_S$ has $6n-12$ triangles and $3n-5$ Steiner points).

Each triangle in $\T_S$ either has exactly one corner or one edge contained in $\partial P$.
The cyclic order in which the triangles appear around $\partial P$ thus induces a Hamiltonian cycle $C$ in $\T_S$.

We define the \emph{wall} of $\T_S$ to be the graph induced by the vertices of $\T_S$, that are not on $\partial P$.
These vertices are the centroids of the triangles in the original triangulation $\T$ and the midpoints of the diagonals of $\T$, so it follows that the wall is a tree.
It is easy to see that each triangle in $\T_S$ either has (i) a single edge on $\partial P$ and a single vertex on the wall, or (ii) a single edge on the wall and a single vertex on $\partial P$.

\begin{figure}
\centering
\includegraphics[page=16]{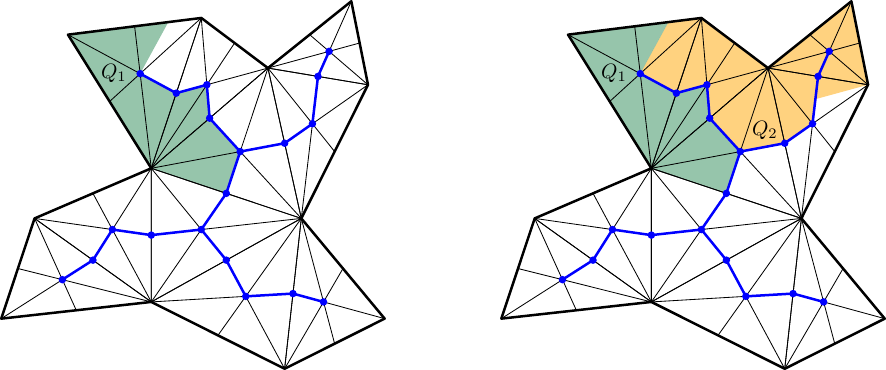}
\caption{The construction of the two first pieces $Q_1$ and $Q_2$.}
\label{fig:areapieces}
\end{figure}

We now construct our pieces in the following way; see \Cref{fig:areapieces}.
Let the triangles of $\T_S$ be $\triangle_1,\ldots,\triangle_m$ in the order of $C$.
We find the first triangle $\triangle_i$ such that $\sum_{j=1}^i \area \triangle_j\geq a_1$.
If $\triangle_i$ has an edge $e$ on $\partial P$, we cut it along a segment from the vertex on the wall to the point on $e$ such that the part incident with $\triangle_{i-1}$ together with the triangles $\triangle_1,\ldots,\triangle_{i-1}$ has area $a_1$, and these triangles (the last of which is just a subset of $\triangle_i$) constitute the first piece $Q_1$. Since $Q_1$ is the union of triangles visited on a path in $C$, then $Q_1$ is a simple polygon.  
The case that $\triangle_i$ has an edge $e$ on the wall and a vertex on $\partial P$ is handled similarly.

In general, when constructing a piece $Q_i$ with area $a_i$, we start at a segment connecting $\partial P$ to the wall, that was introduced when constructing the previous piece $Q_{i-1}$.
We then proceed to traverse the triangles in the order of $C$ until the traversed triangles have total area at least $a_i$ and then make an appropriate cut.

Since we traverse through $O(n)$ triangles and we make $O(k)$ cuts, this approach results in an algorithm with optimal running time $O(n+k)$, proving \Cref{thm:main-1}.

\section{Open questions}

We close the paper by suggesting some directions for future research.
We believe that our techniques can also be adapted to find $O(1)$-approximations of other natural variants of the problem of partitioning a simple polygon into small pieces.
These include partitions into pieces each of which is contained in an equilateral triangle of fixed size (or any other fixed polygon), that can either be allowed to be arbitrarily rotated or not.
Another version is pieces with bounded geodesic radius, which is closely related to the geodesic diameter partitioning problem, but it seems more involved to find a good boundary partition.

New ideas seem to be needed in order to partition polygons with holes (except for the area partitioning problem, where our algorithm does handle polygons with holes).
Here we would also allow the pieces to have holes.
It is interesting whether such algorithms can be found.
Another interesting direction, that, however, seems even more difficult, would be to find algorithms for partitioning three-dimensional (non-convex) polyhedra into small pieces.
This seems especially relevant to 3D printing and other manufacturing processes.
Heuristics have been proposed to solve this problem~\cite{DBLP:journals/tog/LuoBRM12,chen2022skeleton,jiang2017models}, and several videos can be found on YouTube showing how it can be done manually using various software (for instance, search for \href{https://www.youtube.com/results?search_query=splitting+large+part+for+3D+printing}{``splitting large part for 3D printing''}).

It is interesting whether the techniques of Abrahamsen and Stade~\cite{DBLP:journals/corr/abs-2404-09835} can be used to prove NP-hardness of other versions than the aligned square partitions.

Another research direction would be to consider covers instead of partitions.
Some covering problems have been shown to be $\exists\mathbb R$-complete~\cite{DBLP:journals/jacm/AbrahamsenAM22,DBLP:conf/focs/Abrahamsen21}, but these proofs rely on the use of some very long pieces that can cover distant regions of the polygon, so we wonder if some size-constraint covering problems are also $\exists\mathbb R$-complete.

Yet another direction would be to consider a dual version of the partitioning problems:
We specify the number of pieces as part of the input and seek a partition of the polygon $P$ into that number of pieces, where the objective is to minimize the size of the largest piece.
The size would be measured as the smallest enclosing disk or square, the straight or geodesic diameter or the perimeter.
The paper~\cite{DBLP:conf/esa/ArkinD0GMPT20} contains results on problems of this sort, where instead of fixing the number of pieces, we fix the number of cuts we are allowed to make along chords in $P$.



\printbibliography

\end{document}